\documentclass[superscriptaddress,prx,aps,cp,amsmath,amssymb,reprint,floatfix,longbibliography]{revtex4-2}

\usepackage{lmodern}
\usepackage{graphicx}% Include figure files
\usepackage{dcolumn}% Align table columns on decimal point
\usepackage{bm}% bold math
\usepackage[utf8]{inputenc}
\usepackage[T1]{fontenc}
\usepackage{color}
\usepackage[colorlinks,citecolor=blue,linkcolor=blue,urlcolor=blue,bookmarks=false,hypertexnames=true]{hyperref}
\usepackage[inkscapelatex=false]{svg}

\usepackage{verbatim}
\usepackage{braket}
\usepackage{dsfont}
\usepackage{booktabs}
\usepackage{graphicx}
\usepackage[percent]{overpic}
\usepackage{tikz}
\usepackage{xcolor}
\usetikzlibrary{plotmarks}   % for triangle*, square*, diamond*, 

\usepackage{xcolor}
\definecolor{myPurple}{HTML}{8A2BE2}   % RALLY_T
\definecolor{myLightBlue}{HTML}{56B4E9} % dCRAB pwc
\definecolor{myGreen}{HTML}{009E73}     % RALLY_A
\definecolor{orangered}{HTML}{FF4500}

\usepackage{amsthm}

\graphicspath{{figures/}}

\theoremstyle{definition}

\newtheorem{theorem}{Theorem}[section]
\newtheorem*{theorem*}{Theorem}

\begin{document}

\newcommand{\fraunhoferIAF}{\affiliation{Fraunhofer Institute for Applied Solid State Physics IAF, Tullastr. 72, 79108 Freiburg, Germany}}

\newcommand{\unipdPhysics}{%
  \affiliation{Dipartimento di Fisica e Astronomia ``G.~Galilei'', Università di Padova, I-35131 Padova, Italy}%
}

\newcommand{\paduaQTRC}{%
  \affiliation{Padua Quantum Technologies Research Center, Università degli Studi di Padova, I-35131 Padova, Italy}%
}

\newcommand{\INFN}{%
  \affiliation{Istituto Nazionale di Fisica Nucleare (INFN), Sezione di Padova, I-35131 Padova, Italy}%
}

\author{Marco Dall'Ara}\fraunhoferIAF 
\author{Martin Koppenhöfer}\fraunhoferIAF
\author{Florentin Reiter}\fraunhoferIAF
\author{Thomas Wellens}\fraunhoferIAF
\author{Simone Montangero}\unipdPhysics \paduaQTRC \INFN
\author{Walter Hahn}\email{walter.hahn@iaf.fraunhofer.de}
\fraunhoferIAF

\title{Random layers for quantum optimal control with exponential expressivity}

\begin{abstract}
A long-standing problem in quantum optimal control is finding an optimal pulse structure that leads to an efficient exploration of the unitary space with a minimal number of optimization parameters. We solve this problem by constructing parametrized pulse sequences from random-amplitude pulses grouped in layers with one optimization parameter per layer. We show that, when increasing the number of pulses, the resulting random unitaries converge exponentially fast to the uniform Haar-random ensemble, thus providing for an efficient exploration of the unitary space. Grouping the pulses into layers allows for lowering the total number of optimization parameters. We focus on two random-layer (RALLY) methods: In RALLY$_\text{T}$, time durations of the layers are optimized while the pulse amplitudes are randomly chosen beforehand, possibly even from a few discrete values. RALLY$_\text{A}$ optimizes a joint scaling factor of the random pulse amplitudes in each layer. We numerically validate the two methods by applying them to problems of unitary synthesis, ground-state preparation and state transfer in different quantum systems. For all problems considered, both methods approach an information-theoretic lower bound on the number of optimization parameters and outperform other commonly used algorithms. In gradient-free optimization, the RALLY methods are orders of magnitude more accurate with fewer figure-of-merit evaluations. The RALLY methods are promising for advancing quantum machine learning and variational quantum algorithms.
\end{abstract}

\maketitle

\section{Introduction}
Precise and fast control of quantum systems is essential for many problems in quantum science and technology. Quantum-optimal-control theory provides methods to achieve such control by designing pulse sequences that, for example, steer quantum systems to target states or induce target dynamics~\cite{qoc_review_2015,koch-review,Ansel_2024}. These pulse sequences are obtained by optimizing a suitable figure of merit either directly on a quantum device (closed-loop) or using a numerical model of the quantum system (open-loop)~\cite{krotov,GRAPE,CRAB,dcrab,goat,Schäfer_2021, sivak,adhoc,rand_benchmark,rabitz_review, NMR_hyp,qoc_phase_transition,opc_res1,dolde_2014,waldherr,sensing_neu,Ryd_sim_GHZ,Kokail_2019}.

With increasing size of controllable quantum systems, several conceptual challenges emerge in quantum optimal control. Finding optimal control in increasingly large systems while retaining high accuracy requires the number of optimization parameters to scale with the Hilbert-space dimension. This complicates the optimization of parameters and leads to an exponential scaling of computational resources for open-loop optimization. Including experimental constraints makes this problem even more difficult. When restricting pulse amplitudes to take only certain discrete values, which is interesting for cases of limited experimental capabilities and for imposing structures on the pulse sequences, the optimization becomes substantially more challenging~\cite{PMP,dridi_discrete}. This includes bang–bang or dynamical-decoupling structures~\cite{DD,bangbang}, which reduce the coupling to the environment and which were shown to correspond to time-optimal control in special cases~\cite{bang_off2,single-spin}. These challenges call for novel control-optimization schemes for enabling both current and future generations of quantum experiments and devices.

In this article, we introduce a family of efficient parametrized pulse sequences that address the above challenges and solves the long-standing problem of finding an optimal pulse structure for quantum optimal control~\cite{dcrab,dCRAB_bases}.
The pulse sequences are assembled in $N_\text{L}$ layers, each of which is a sequence of $N_\text{P}$ (layer size) random constant-amplitude pulses characterized by one optimization parameter. 
We focus on two particular powerful random-layer (RALLY) methods: RALLY$_\text{T}$ and RALLY$_\text{A}$, illustrated in Fig.~\ref{fig:RALLY_illustration}. RALLY$_\text{T}$ optimizes the time durations of the layers keeping the random amplitudes fixed, while RALLY$_\text{A}$ optimizes a joint factor that scales the randomly chosen amplitudes in a layer keeping the pulse durations fixed. We first rigorously show that, upon sampling over the RALLY optimization parameters \textit{and} control amplitudes, the ensemble of unitaries reachable by this pulse-sequence structure converges exponentially fast with $N_\text{P}N_\text{L}$ to the Haar-random ensemble~\cite{Haar_convergence}. This provides exponential expressivity, i.e., efficient exploration of the entire unitary space. However, when treating all those random variables as optimization parameters to solve a generic task, this leads to a challenging optimization problem in practice. This motivates our layered pulse-sequence structure, where the individual control amplitudes are not directly optimized and groups of $N_\text{P}$ pulses share a joint parameter. 
We provide compelling numerical evidence that, for the overwhelming majority of \textit{individual random choices of control amplitudes}, the ensemble of reachable unitaries exhibits the same convergence behavior with $N_\text{P}N_\text{L}$ even though we sample only $N_\text{L}$ RALLY parameters. This feature of grouping pulses is key to ensure that the expressivity of our pulse sequences increases with the layer size $N_\text{P}$ without changing the number of optimization parameters $N_\text{L}$, and enables efficient optimization. For RALLY$_\text{T}$, this even holds when the pulse amplitudes are randomly chosen from a discrete set of values. RALLY$_\text{A}$, in turn, can be neatly incorporated into the existing CRAB~\cite{CRAB} and dCRAB~\cite{dcrab} algorithms by interpreting it as a new enhanced set of basis functions.

Beyond exponential expressivity, the RALLY$_\text{T}$ method provides other important advantages. Constraints of the experimental setting, such as a finite bandwidth, are readily accommodated. Further, an additional penalty term on the total time duration can be introduced to reduce the total pulse-sequence duration without compromising optimization efficiency. RALLY$_\text{T}$ is computationally efficient to simulate making it applicable to larger systems than common quantum-optimal-control algorithms. Both methods RALLY$_\text{T}$ and RALLY$_\text{A}$ can be used in gradient-based and gradient-free optimization, hence they can be applied in closed-loop and open-loop pulse optimization. 

To demonstrate the performance of the RALLY methods, we perform numerical simulations for three problems: (i) unitary synthesis of a three-qubit gate and (ii) ground-state preparation on a globally driven Rydberg-atom platform, and (iii) state transfer in an Ising spin chain. For these problems, the RALLY methods approach an information-theoretic lower bound on the number of optimization parameters~\cite{Control_landscape_Rabitz,Lloyd_2014} and exhibit clear advantages when compared to other commonly used algorithms. In particular, we observe that, in gradient-free optimization, both RALLY methods achieve several orders of magnitude higher accuracies with fewer figure-of-merit evaluations. We demonstrate that, with minor modifications, RALLY$_\text{T}$ yields pulse sequences with reduced total time. By using a gradient-based optimization routine and randomly selecting pulse amplitudes from two discrete values, RALLY$_\text{T}$ accurately optimizes systems with two more qubits than the GRAPE algorithm~\cite{GRAPE} within the same time budget, while also accommodating realistic bandwidth limitations. 

In the next section (Sec.~\ref{sec_rally}), we introduce the RALLY methods and discuss their advantages. Section~\ref{sec_res} presents benchmarks of the RALLY methods. We conclude in Sec.~\ref{sec_conclude}.

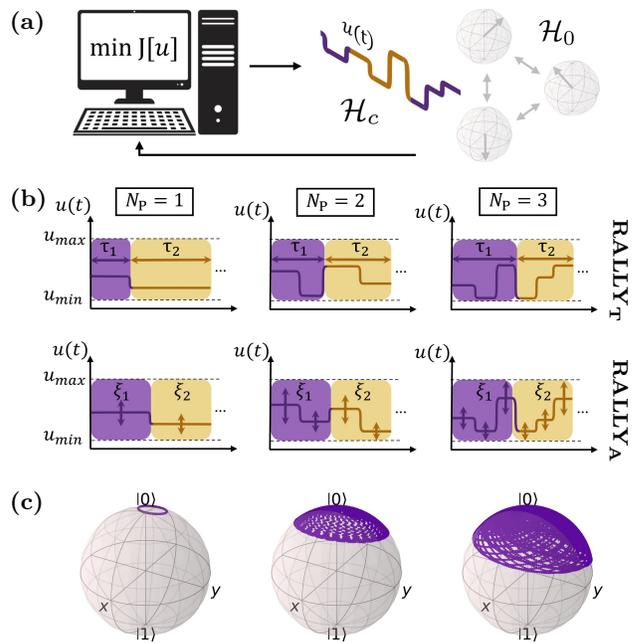
\begin{figure}[t]
  \includegraphics[width=0.4\textwidth]{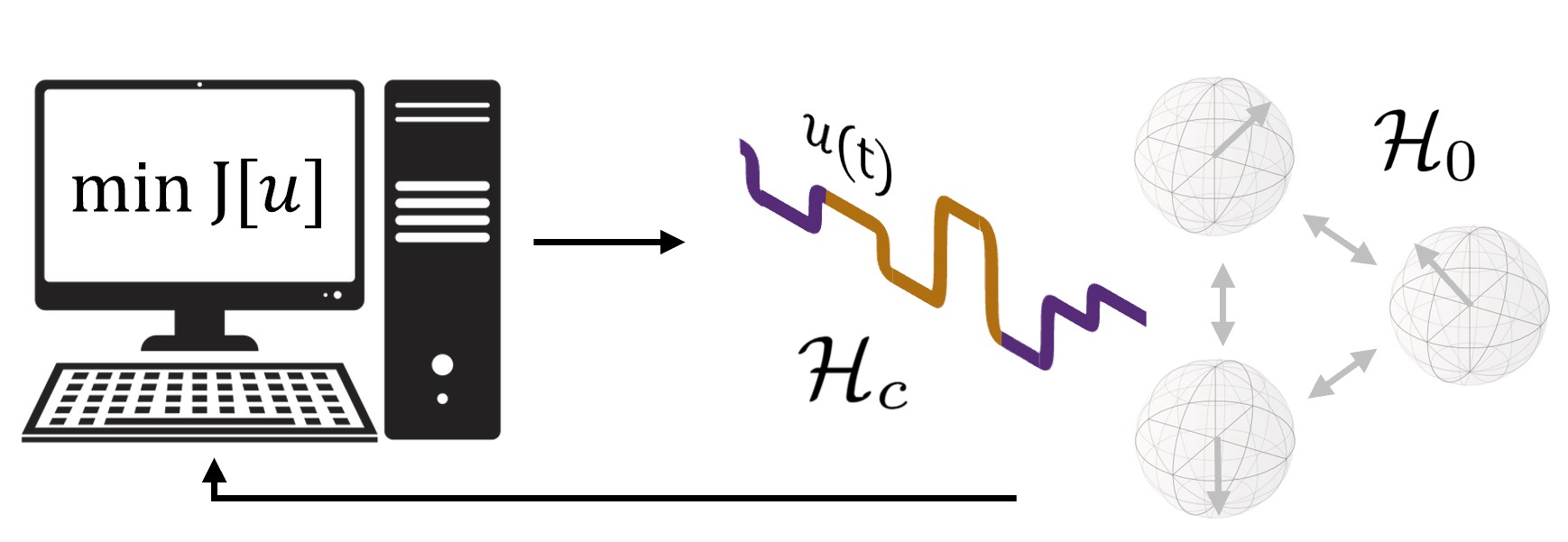}\par
  \vspace{0.5em} % adjust the gap (e.g., 1em, 10pt, 1cm)
  \includegraphics[width=0.44\textwidth]{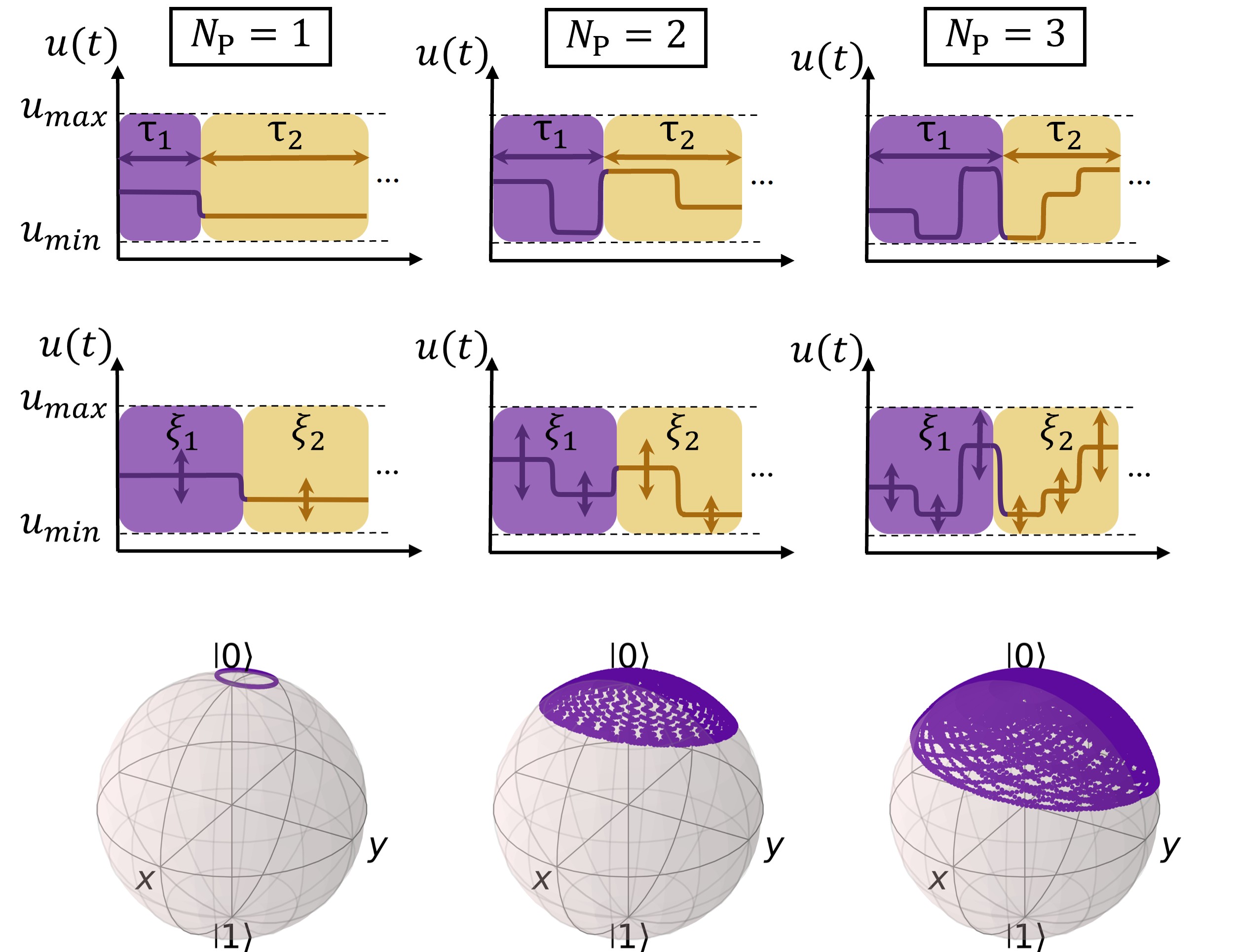}
  \begin{picture}(0,0)
    \put(-237,232){\textbf{(a)}}
    \put(-237,165){\textbf{(b)}}
    \put(-10,160){\rotatebox{-90}{\footnotesize\bfseries RALLY$_\text{T}$}}    
    % \put(-237,115){\textbf{(c)}}
        \put(-10,107){\rotatebox{-90}{\footnotesize\bfseries RALLY$_\text{A}$}}    \put(-237,50){\textbf{(c)}}
    
  \end{picture}

  \caption{
    \textbf{(a)} Sketch of a pulse-optimization procedure in quantum optimal control. The dynamics of a quantum system with internal Hamiltonian ${\cal H}_{0}$ is controlled by applying a pulse sequence $u(t)$ that couples to the system via the control Hamiltonian ${\cal H}_{c}$. The optimal pulse is found by minimizing a figure of merit $J[u]$ in a feedback loop. \textbf{(b)} In RALLY, the control pulse sequence $u(t)$ is divided into $N_\text{L}$ layers (different colors) with $N_\text{P}$ pulses per layer. The pulse amplitudes are chosen randomly. In RALLY$_\text{T}$, the layer durations \(\tau_\ell\) are optimized. In RALLY$_\text{A}$, the joint scaling factors \(\xi_\ell\) of the control amplitudes are optimized in each layer. \textbf{(c)} Illustration of how the expressivity increases for larger $N_\text{P}$. When changing $\tau$, the dynamics of a single layer traces out a trajectory on the Bloch sphere. For larger $N_\text{P}$, this trajectory becomes more complex and covers an increasing portion of the Bloch sphere. 
    Single-spin trajectories are shown for the Hamiltonian $\mathcal{H}(t) = 5\sigma_z + u(t)\sigma_x$, initial state $\lvert 0 \rangle$, evolution time $\tau = 100$ and $u(t)$ composed of fixed control amplitudes whose values are randomly chosen from the interval $[-2,2]$.
  }
  \label{fig:RALLY_illustration}
\end{figure}

\section{Random-Layer methods} \label{sec_rally}

We consider a quantum system described by a Hamiltonian operator of the form ${\cal H}(t)={\cal H}_{0}+u(t){\cal H}_{c}$, where the internal Hamiltonian ${\cal H}_0$ is usually referred to as the drift Hamiltonian and ${\cal H}_{c}$ is the control Hamiltonian with the corresponding time-dependent control field $u(t)$~\footnote{The results of this article are readily extended to the general case of many control fields $u_k(t)$ with corresponding control Hamiltonians ${\cal H}_{k}$ such that the total Hamiltonian reads ${\cal H}(t)={\cal H}_{0}+\sum_{k}u_{k}(t){\cal H}_{k}$.}. We define a pulse as the evolution over a period of time with a constant control field.

\subsection{Definition of the RALLY pulse sequences}
The pulses in a RALLY pulse sequence are arranged in layers
% \begin{equation}
    $U_{\text{R}}(\theta_1, .., \theta_{N_\text{L}})=
     \prod_{\ell=1}^{N_\text{L}}
     U_{\ell} (\theta_\ell),$
% \end{equation}
where $U_{\ell} (\theta_\ell)$ is the propagator describing layer $\ell$ with optimization parameter $\theta_\ell$ and $N_\text{L}$ is the total number of layers. For each layer, $U_{\ell}(\theta_\ell)$ is a series of $N_\text{P}$ constant-amplitude pulses (layer size), where the pulse amplitudes are chosen randomly from a predefined set.

We propose two RALLY methods with optimization parameters being either layer durations ($\theta_\ell=\tau_\ell$) in RALLY$_\text{T}$ or scaling factors ($\theta_\ell=\xi_\ell$) for control amplitudes in RALLY$_\text{A}$, cf. Fig.~\ref{fig:RALLY_illustration}. We denote the corresponding propagators as $U_{\text{R,T}}(\tau_1, .., \tau_{N_\text{L}})$ and $U_{\text{R,A}}(\xi_1, .., \xi_{N_\text{L}})$, respectively. For $U_{\text{R,T}}(\tau_1, .., \tau_{N_\text{L}})$, we define 
\begin{equation} \label{eq:RALLY_T}
U_{\text{R,T}}(\bm{\tau})=
    \prod_{\ell=1}^{N_\text{L}}
     \prod_{p=1}^{N_\text{P}}
     \exp\biggl[-i\frac{\tau_\ell}{N_\text{P}}\bigl({\cal H}_{0} + u^{(\ell,p)}{\cal H}_{c}\bigr)\biggr],
\end{equation}
where $\bm{\tau} = (\tau_1, ..., \tau_{N_\text{L}})$ and pulse durations $\tau_\ell/N_\text{P}$ are chosen to be the same for all pulses in a layer for simplicity~\footnote{In general, the individual pulse durations within a layer could be chosen as arbitrary fractions of $\tau_\ell$.}. For $U_{\text{R,A}}(\xi_1, .., \xi_{N_\text{L}})$, we define
\begin{equation} \label{eq:RALLY_AMP}
U_{\text{R,A}}(\bm{\xi})=
    \prod_{\ell=1}^{N_\text{L}}
     \prod_{p=1}^{N_\text{P}}
     \exp\biggl[-i\Delta t\bigl({\cal H}_{0} + \xi_\ell\,u^{(\ell,p)}{\cal H}_{c}\bigr)\biggr],
\end{equation}
where $\bm{\xi} = (\xi_1, ..., \xi_{N_\text{L}})$ and all pulses have the same time duration $\Delta t$ for simplicity. Importantly, for both pulse sequences, the amplitudes of the individual pulses $u^{(\ell,p)}$ are randomly chosen. Note that, due to the presence of the drift Hamiltonian $\mathcal{H}_0$, the two propagators $U_{\text{R,T}}(\bm{\tau})$ and $U_{\text{R,A}}(\bm{\xi})$ are generally different. Both propagators depend on the control-field amplitudes $u^{(\ell,p)}$ but we indicate this dependence only when necessary.

The RALLY$_\text{A}$ propagator $U_{\text{R,A}}(\bm{\xi})$ can be viewed as a new basis choice for the CRAB~\cite{CRAB} or dCRAB~\cite{dcrab} algorithms with control field $u(t)=\sum_\ell \xi_{\ell}\,\phi_{\ell}(t)$, where each basis function $\phi_{\ell}(t)$ is a concatenation of piecewise-constant segments with amplitudes $\{u^{(\ell,p)}\}_p$; see Appendix~\ref{app:dCRAB_GRAPE} for details of the dCRAB algorithm.

\subsection{Convergence properties of the RALLY pulse sequences}
To investigate the convergence properties of the RALLY pulse sequences, let us consider the ensemble of unitary operators $U_{\text{R}}(\bm{\theta})$ obtained by sampling the optimization parameters $\bm{\theta}$ (with $\bm{\theta}=\bm{\tau}$ for $U_{\text{R,T}}(\bm{\tau})$ and $\bm{\theta}=\bm{\xi}$ for $U_{\text{R,A}}(\bm{\xi})$). The main question is whether this ensemble of unitaries approximates a Haar-random ensemble on $U(N)$, which corresponds to a uniform distribution in the unitary group. Exhibiting properties of the Haar ensemble with respect to the optimization parameters entails a uniform coverage of the unitary group, thereby making the optimization of an arbitrary target unitary more efficient. 

In the following, we first discuss and analytically prove a theorem showing that the ensemble of reachable unitaries approaches the Haar ensemble exponentially with $N_\text{L} N_\text{P}$ when, in addition to the sampling of optimization parameters $\bm{\theta}$ of the RALLY pulse sequences, the control amplitudes \(u^{(\ell,p)}\) are also sampled. From a practical perspective, this corresponds to evaluating the expressivity of a pulse sequence, when both $\bm{\theta}$ and \(u^{(\ell,p)}\) are optimized. However, the random control amplitudes \(u^{(\ell,p)}\) remain unchanged in the RALLY pulse sequences. We provide compelling numerical evidence that the conclusions of the above theorem hold even for almost any fixed random choice of \(u^{(\ell,p)}\), thereby arriving at the setting of the RALLY methods. This property of randomly chosen unitaries is similar to the typicality of randomly chosen quantum states, e.g., canonical typicality~\cite{gemmer_mahler,can_typ,pop}.

Let us start by considering the following important theorem: For a fully controllable quantum system, the ensemble of unitaries reachable by the RALLY propagator $U_{\text{R}}(\bm{\theta})$, where \textit{both} optimization parameters $\bm{\theta}$ and control amplitudes $u^{(\ell,p)}$ are sampled from a continuous distribution,  converges exponentially fast in the number of pulses $N_\text{L} N_\text{P}$ to the uniform Haar ensemble. Moreover, if the control amplitudes are drawn from a set of discrete values, the theorem holds by considering a weaker convergence criterion~\cite{Haar_convergence}. For a rigorous treatment and a proof, we refer the reader to Appendix~\ref{app:t-unitary convergence}, and for the definition of full controllability to Appendix~\ref{app:cond}.

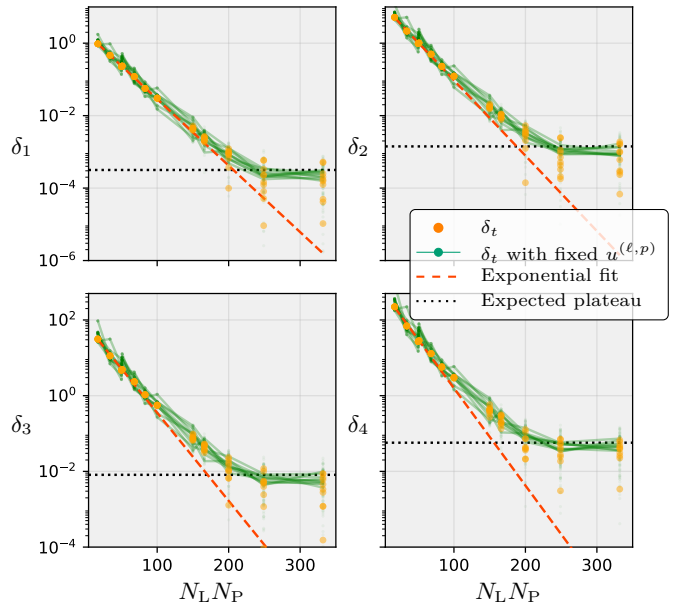
\begin{figure}[b]
  \centering
  \begin{overpic}[width=0.95\columnwidth]{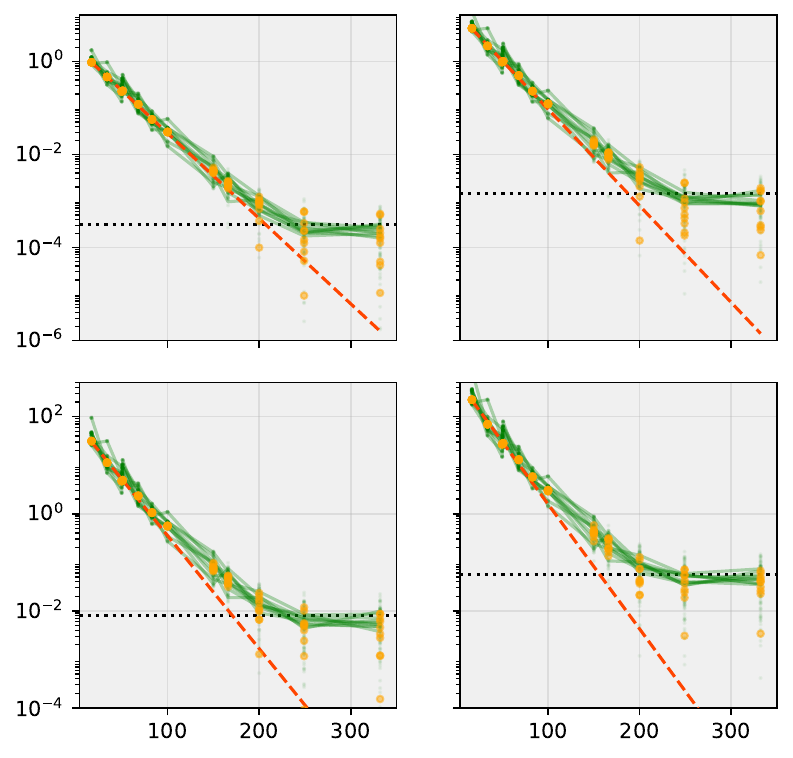}
    % Common x-axis label (centered at bottom)
    \put(30,-1){\makebox(0,0){\small $N_\text{L}N_\text{P}$}}
    \put(78,-1){\makebox(0,0){\small $N_\text{L}N_\text{P}$}}

    \put(-2.5,70){{\small $\delta_1$}}
    \put(52,70){{\small $\delta_2$}}
    \put(-2.5,25){{\small $\delta_3$}}
    \put(52,25){{\small $\delta_4$}}

    \put(62,43){%
      \begin{tikzpicture}
        \node[
          fill=white,
          fill opacity=0.8,
          text opacity=1,
          draw=black,
          rounded corners=2pt,
          inner sep=2pt
        ] {
          \begin{tikzpicture}[x=1em,y=1em]
    
            % RALLY_A: green, line + square
            % \draw[orange, thick, dashed] (0,2) -- (1.8,2);
            \draw[orange, thick,
                  mark=*,
                  mark size=1.5pt, 
                  mark options={solid}]
              plot coordinates {(0.9,3)};
            \node[right=0.6em] at (1.8,3) {\scriptsize $\delta_t$};

            \draw[myGreen] (0,2) -- (1.8,2);
            \draw[myGreen,
                  mark=*,
                  mark size=1.5pt, 
                  mark options={solid}]
              plot coordinates {(0.9,2)};
            \node[right=0.6em] at (1.8,2) {\scriptsize $\delta_t$ with fixed $u^{(\ell, p)}$};

            \draw[orangered, thick, dashed] (0,1) -- (1.8,1);
            \node[right=0.6em] at (1.8,1) {\scriptsize Exponential fit};

            \draw[black, thick, dotted] (0,0) -- (1.8,0);
            \node[right=0.6em] at (1.8,0) {\scriptsize Expected plateau};
          \end{tikzpicture}
        };
      \end{tikzpicture}%
    }

  \end{overpic}
  \caption{Exponential convergence of the first four moments of the ensemble of unitaries reachable by the RALLY$_\text{T}$ propagator $U_\text{R,T}$ to the corresponding moments of a Haar-random distribution. The difference $\delta_t$~\eqref{eq:delta_t} between the corresponding moments is shown for $t=1,2,3,4$ as a function of $N_{\mathrm{L}}N_{\mathrm{P}}$. Orange points show $\delta_t$ with layer durations drawn uniformly from $[0,10]$ and control amplitudes drawn uniformly from $[-1,1]$. The dashed red line is a linear fit on the logarithmic scale. Importantly, results for $\delta_t$ follow the same decay curve also when \textit{only layer durations} are sampled and the control amplitudes are fixed (shown by the 10 green lines). We attribute the fluctuations of the green lines to the finite system sizes; see Appendix~\ref{app_typ}. The plateaus at larger values of $N_{\mathrm{L}}N_{\mathrm{P}}$ are due to the finite sample size and they are consistent with the expected values from the Haar distribution (black dotted lines); see Appendix~\ref{App:Harr_conv} for details. For each value of $N_{\mathrm{L}}N_{\mathrm{P}}$, where $N_{\mathrm{L}}\in\{17, 50, 83\}$ and $N_{\mathrm{P}}\in\{1,2,3,4\}$, we use $10^7$ samples and repeat the calculation 10 times. One set of randomly chosen values of $h_{ix}$ and $h_{iz}$ for the Hamiltonian ${\cal H}_\text{I}$~\eqref{eq:random_spin_H} was used. }
  \label{fig:conv_plot}
\end{figure}

We corroborate the above theorem with numerical simulations of an Ising spin chain consisting of \(n=3\) \mbox{spins-1/2} in a longitudinal and transverse magnetic field
\begin{equation} \label{eq:random_spin_H}
{\cal H}_\text{I} =
\sum_{i=1}^{n} h_{ix} \sigma_{ix}+\sum_{i=1}^{n} h_{iz} \sigma_{iz}+J(t)\sum_{i=1}^{n-1}\sigma_{ix}\sigma_{(i+1)x},
\end{equation}
where $h_{ix}$ ($h_{iz}$) are the spin-dependent magnetic-field components along the $x$-axis ($z$-axis), $\sigma_{i\alpha}$ is the Pauli operator of spin $i$ along direction $\alpha$ and the interaction coefficient $J(t)$ is the time-dependent control field. The components $h_{ix}$ and $h_{iz}$ are randomly chosen from the interval \([0.5,1]\), which ensures the absence of symmetries and full controllability~\cite{dcrab}. We present the analysis for RALLY$_\text{T}$, but based on the optimization performance described below, we anticipate similar convergence behavior for RALLY$_\mathrm{A}$ and for RALLY$_\mathrm{T}$ with control amplitudes drawn from a discrete set of values.

We investigate the convergence of the ensemble of reachable unitaries to the Haar ensemble by studying the convergence of the first four moments of the corresponding distribution, which is equivalent to convergence to unitary designs of order up to 4. The $t$-th moment converges if and only if the following expression vanishes~\cite{Scott_2008}
\begin{equation}
\label{eq:delta_t}
    \delta_t
    = \bigg|
      \mathbb{E}_{x,y \sim \nu}
      \bigg[
        \bigl|\operatorname{tr}\bigl(U_{\mathrm{R}}(x)^\dagger U_{\mathrm{R}}(y)\bigr)\bigr|^{2t}
      \bigg]
      - t!
      \bigg|,
\end{equation}
where the expectation is taken over two independent samples $x,y = (\bm{\tau}, u^{(\ell,p)}) \sim \nu$ with $\nu$ being the joint probability distribution of layer durations and control amplitudes. The numerical results shown in Fig.~\ref{fig:conv_plot} exhibit an exponential decay of $\delta_t$ for the moments $t\leq4$ with $N_{\mathrm{L}} N_{\mathrm{P}}$, which is in accordance with the theorem.

Remarkably, we observe the same exponential decay also for each \textit{individual random choice} of the control amplitudes as also shown in Fig.~\ref{fig:conv_plot}. At first sight, an exponential decay with $N_\text{L}$ is anticipated because only $N_\text{L}$ layer durations are sampled. However, the typical behavior observed for each individual random choice of the control amplitudes implies that each individual pulse entering a RALLY pulse sequence contributes equally to the expressivity of the corresponding propagator, hence the exponential decay with $N_\text{L}N_\text{P}$. We attribute the small fluctuations of $\delta_t$ for the individual random choices of the control amplitudes on top of the exponential decay to the finite size of the quantum system. When reducing the number of spins in the Ising spin chain, these fluctuations become larger; see Appendix~\ref{app_typ}. This typical behavior provides the foundation for the efficacy of the RALLY methods.

From the above results follows that the convergence properties of the RALLY pulse sequences can be enhanced by increasing the layer size $N_\text{P}$ without changing the number of optimization parameters $N_\text{L}$. Thereby, over-parametrization can be avoided at the expense of a larger total number of pulses (see Appendix~\ref{app:cond} for details on over-parametrization). Qualitatively, this can be understood on the basis of a single-spin dynamics, cf. Fig.~\ref{fig:RALLY_illustration}c. For $N_\text{P}=1$, the spin dynamics within a layer is confined to a circle on the Bloch sphere for any random choice of the pulse amplitude. For $N_\text{P}>1$, the dynamics is more complicated because the effective Hamiltonian ${\cal H}_\text{eff}$ for a layer is time-dependent due to the series of $N_\text{P}$ rotations around different axes: $U_{\text{R,T}}(\bm{\tau})=\prod_{\ell=1}^{N_\text{L}}\exp[-i\tau_\ell{\cal H}_\text{eff}(\tau_\ell)]$ for RALLY$_\text{T}$, cf. Eq.~\eqref{eq:RALLY_T}. Therefore, the corresponding path covers a larger region of the Bloch sphere in a loose analogy to space-filling curves that explore a two-dimensional manifold. This corresponds to a more efficient exploration of the unitary group by $U_{\text{R}}$. The above intuition naturally extends to higher-dimensional settings.

\subsection{Advantages of the RALLY$_\text{T}$ method}
One of the key features of the RALLY$_\text{T}$ method is that the pulse amplitudes $u^{(\ell,p)}$ are randomly chosen \emph{prior} to the actual optimization and remain unchanged thereafter. Therefore, by tailoring the predefined set from which the coefficients \(u^{(\ell,p)}\) are drawn, one can explicitly enforce hardware constraints, such as a maximum-amplitude condition \(\lvert u^{(\ell,p)}\rvert \le u_{\max}\), or restrict the pulse amplitudes to a discrete set of allowed values. The latter is interesting, for example, for cases of limited experimental capabilities and for imposing particular structures on the pulse sequences. For example, dynamical-decoupling structures~\cite{DD} are known to reduce the coupling to the environment. Further, pulse sequences similar to bang-bang $u^{(\ell,p)}\!\in\!\{\pm u_{\max}\}$~\cite{bangbang} and bang-singular-bang protocols $u^{(\ell,p)}\in\{0,u_{\max}\}$~\cite{bang_off2} were shown to be time optimal for the rotation of a single spin-1/2~\cite{single-spin}.

\begin{figure}[b]
  \centering
    \begin{overpic}[width=0.31\textwidth]{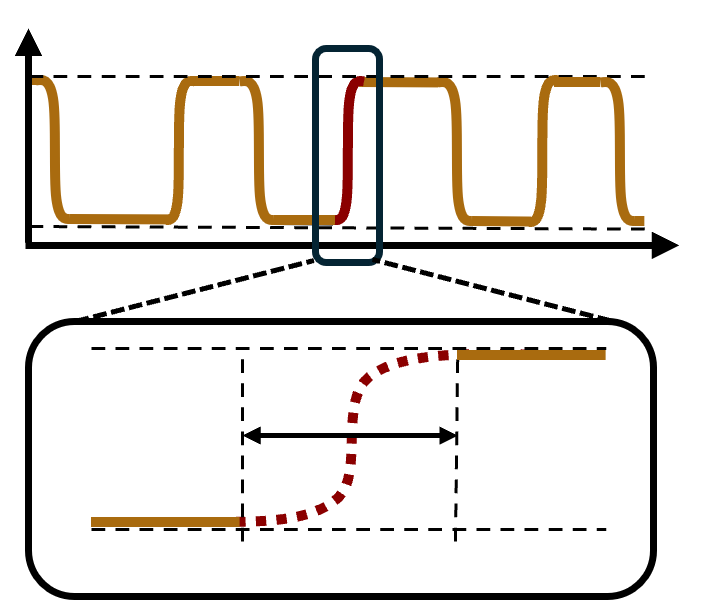}

    % Example labels/formulas:
    \put(-3,80){$u$}
    \put(98,48){$t$}
    \put(-10,73){$u_{\text{max}}$}
    \put(-10,53){$u_{\text{min}}$}
    \put(-3,80){$u$}
    \put(39,28){$\tau_{\text{rise}}$}
    \put(45,4){$U_{\text{rise}}$}

  \end{overpic}
  \caption{
   Illustration of the interpolation between successive control amplitudes to accommodate the hardware's finite bandwidth in RALLY$_\text{T}$. A limited bandwidth $\Delta\Omega$ imposes a finite rise time $\tau_\text{rise}$, which is implemented using an interpolating function that depends on the amplitudes but not on the pulse durations. Since RALLY$_\text{T}$ samples amplitudes before the optimization and optimizes time durations, the propagators for these interpolations denoted by $U_{\text{rise}}$ can be precomputed and reused during optimization.  
  }
  \label{fig:bandwidth}
\end{figure}
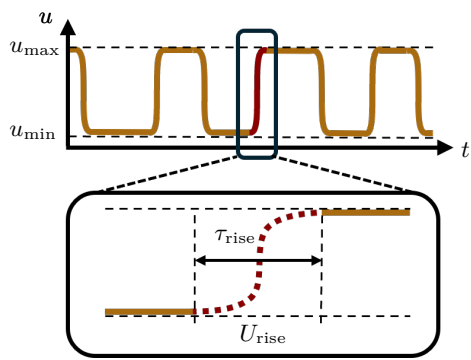

Optimizing layer durations enables resource-efficient simulations: For the RALLY$_\text{T}$ propagator $U_\text{R,T}$~\eqref{eq:RALLY_T}, the Hamiltonians ${\cal H}^{(\ell,p)} = {\cal H}_{0} +  u^{(\ell,p)}{\cal H}_{c}$ can be diagonalized in advance, such that the figure of merit and its gradient can be evaluated across parameter values using only matrix–matrix or matrix-vector multiplications. %Marco: Removed the changes
This is particularly advantageous when \(u^{(\ell,p)}\) are drawn from a small discrete set. This contrasts with methods that require repeated exponentiation of matrices or matrix diagonalizations for each figure-of-merit evaluation.

Finite-bandwidth constraints, which typically pose challenges for piecewise-constant pulse sequences, are readily accommodated in RALLY$_\text{T}$. An interpolation between the pulse amplitudes, that reflects hardware limitations, can be introduced as sketched in Fig.~\ref{fig:bandwidth}. Since the pulse amplitudes are chosen prior to the optimization and the unitary evolutions describing the interpolations do not depend on the pulse durations, the unitary evolutions can be numerically calculated in advance and reused during the optimization with negligible computational overhead.

An additional penalty term (for example, of the form $[\sum_{\ell=1}^{N_\text{L}} \tau_\ell]^2$) allows for reducing the total pulse-sequence duration along with the actual optimization. For quantum devices, which restrict changes of the control amplitudes to periodic grid points in time~\cite{PaganoErrorBudget}, the optimized layer durations can be rounded to the nearest point in time provided that the solution found is sufficiently stable with respect to small perturbations in time; see Appendix~\ref{app:roboustness}. A subsequent fine-tuning optimization with RALLY$_\text{A}$ for the control amplitudes could also be performed.

\section{Numerical simulations and benchmarking} \label{sec_res}
We demonstrate the advantages of the RALLY methods by means of numerical simulations for three paradigmatic control tasks: unitary synthesis, ground-state preparation, and state transfer. For unitary synthesis and ground-state preparation, we consider a globally driven Rydberg-atom platform, and for state transfer, an Ising spin chain. To benchmark the RALLY methods, we compare their performance to that of the dCRAB and GRAPE algorithms. Details of the latter algorithms are provided in Appendix~\ref{app:dCRAB_GRAPE}.

\subsection{Unitary synthesis of a three-qubit gate} \label{sec:synthesis}
We now focus on the unitary synthesis on a globally-driven Rydberg-atom simulator platform with $n$ atoms (qubits), for which the Hamiltonian reads~\cite{rydberg}
\begin{align}
\label{eq:analog_sim}
    {\cal H}_\text{R}(\Omega, \Delta) &= {\cal H}_\text{0} + \Omega {\cal H}_\text{x} - \Delta {\cal H}_\text{z}\\
    &= \sum_{i<j}^n J_{ij} n_i n_j+ \Omega \sum_{i=1}^n  \sigma_{ix} - \Delta \sum_{i=1}^n \sigma_{iz},\nonumber
\end{align}
where \(J_{ij}\cong 1/r^6_{ij}\) describes the Van der Waals interaction between atoms $i$ and $j$ with \(r_{ij}\) being the interatomic distance, $n_i=(\mathbf{1}-\sigma_{iz})/2$ is the occupation operator, and \(\Omega\) and \(\Delta\) are the Rabi frequency and the detuning of the external global driving, respectively. The positions of the Rydberg atoms are chosen such that the Hamiltonian ${\cal H}_\text{R}(\Omega, \Delta)$ in Eq.~\eqref{eq:analog_sim} does not exhibit any symmetries (see Appendix~\ref{app:CNOT_coordinates}) and the system is fully controllable, as we verified by means of dynamical-Lie-algebra analysis. The positions of the Rydberg atoms remain unchanged during the dynamics. We set the global Rabi frequency to \(\Omega=1\ \mathrm{MHz}\) for the entire pulse sequence; hence, the detuning \(\Delta\) is the only control field [${\cal H}_\text{R}(\Omega, \Delta)={\cal H}_\text{R}(\Delta)$] varied within the hardware limits (see Appendix~\ref{app:constraints}).

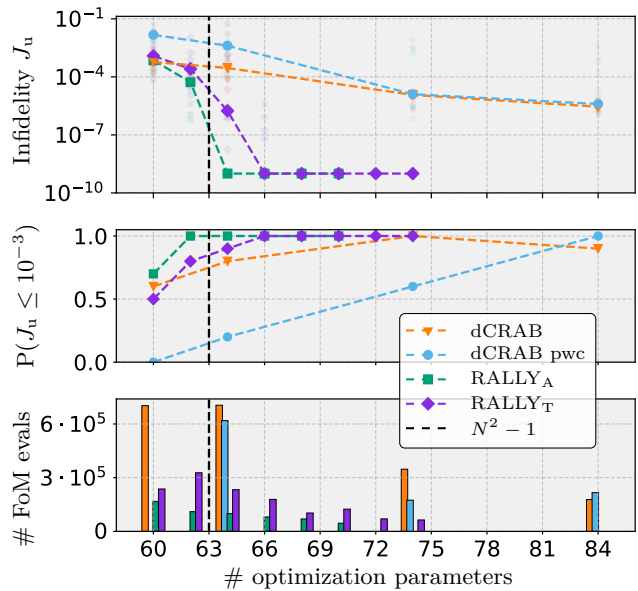
\begin{figure}
  \centering
  \begin{overpic}[width=0.95\columnwidth]{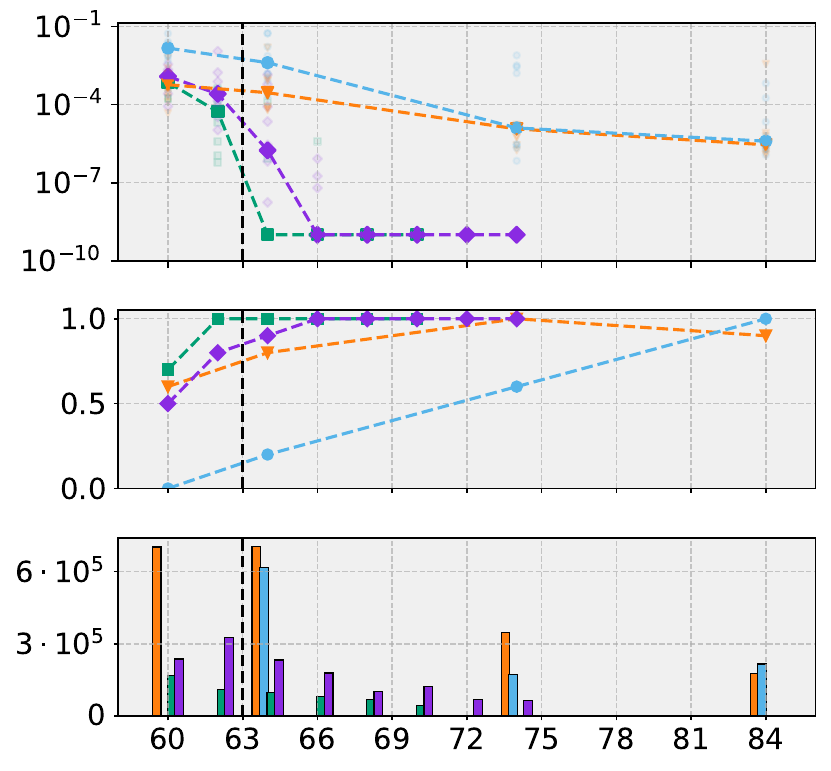}
    % Common x-axis label (centered at bottom)
    \put(55,-1){\makebox(0,0){\small $\#$ optimization parameters}}

    \put(-2.5,66){\rotatebox{90}{\small Infidelity $J_\text{u}$}}
    \put(-2.5,33){\rotatebox{90}{\small P$(J_\text{u}\leq 10^{-3})$}}
    \put(-2.5,4){\rotatebox{90}{\small $\#$ FoM evals}}

    \put(60,20){%
      \begin{tikzpicture}
        \node[
          fill=white,
          fill opacity=0.8,
          text opacity=1,
          draw=black,
          rounded corners=2pt,
          inner sep=2pt
        ] {
          \begin{tikzpicture}[x=1em,y=1em]
            % dCRAB Fourier: orange, line + triangle down
            \draw[orange, thick, dashed] (0,4) -- (1.8,4);
            \draw[orange, thick,
                  mark=triangle*,
                  mark size=1.5pt, 
                  mark options={solid, fill=orange, rotate=180}]
              plot coordinates {(0.9,4)};
            \node[right=0.6em] at (1.8,4) {\scriptsize dCRAB};
    
            % dCRAB pwc: light blue, line + circle
            \draw[myLightBlue, thick, dashed] (0,3) -- (1.8,3);
            \draw[myLightBlue, thick,
                  mark=*,
                  mark size=1.5pt, 
                  mark options={solid}]
              plot coordinates {(0.9,3)};
            \node[right=0.6em] at (1.8,3) {\scriptsize dCRAB pwc};
    
            % RALLY_A: green, line + square
            \draw[myGreen, thick, dashed] (0,2) -- (1.8,2);
            \draw[myGreen, thick,
                  mark=square*,
                  mark size=1.5pt, 
                  mark options={solid}]
              plot coordinates {(0.9,2)};
            \node[right=0.6em] at (1.8,2) {\scriptsize RALLY$_\text{A}$};
    
            % RALLY_T: purple, line + diamond
            \draw[myPurple, thick, dashed] (0,1) -- (1.8,1);
            \draw[myPurple, thick,
                  mark=square*,
                  mark size=1.5pt, 
                  mark options={solid, rotate=45}]
              plot coordinates {(0.9,1)};
            \node[right=0.6em] at (1.8,1) {\scriptsize RALLY$_\text{T}$};
    
            % N^2 - 1: black line only
            \draw[black, thick, dashed] (0,0) -- (1.8,0);
            \node[right=0.6em] at (1.8,0) {\scriptsize $N^2 - 1$};
          \end{tikzpicture}
        };
      \end{tikzpicture}%
    }

  \end{overpic}
  \caption{
    Performance of RALLY methods for the unitary synthesis of a 3-qubit gate~\eqref{eq:target_unitary} in a globally-driven Rydberg-atom platform. As a function of the number of optimization parameters (\(N_\text{L}\) for RALLY), the figure shows the median unitary infidelity $J_\text{u}$~\eqref{eqn_unitary} (top), success probability for achieving \(J_\text{u}\le 10^{-3}\) over 10 optimization runs (middle) and the median number of figure-of-merit evaluations required to reach \(J_\text{u}\le 10^{-3}\) (bottom). The RALLY methods with \(N_\text{P}=5\) rapidly converge to the target precision of $10^{-9}$ and efficiently saturate the information-theoretic lower bound (vertical dashed line; see Appendix~\ref{app:cond}) on the number of optimization parameters, whereas dCRAB with Fourier and piecewise constant (pwc) bases converges more slowly. The RALLY methods also need fewer figure-of-merit (FoM) evaluations than dCRAB. Missing bars in the bottom panel correspond to cases where no successful optimizations were achieved. All methods use the same optimizer (adaptive Nelder–Mead) and identical stopping criteria. 
  }
  \label{fig:CNOT}
\end{figure}

For the unitary-synthesis problem, we consider the figure of merit~\cite{vetter}
\begin{equation} \label{eqn_unitary}
    J_\text{u}[u]\equiv1-\frac{1}{2^{2n}}\Bigl|\operatorname{Tr}\bigl(U_{\text{T}}^{\dagger}U[u]\bigr)\Bigr|^{2},
\end{equation}
which is sometimes referred to as unitary infidelity. The target unitary is the three-qubit gate
\begin{equation} \label{eq:target_unitary}
  U_{\text{T}} = \mathrm{CNOT}_{1,2}\,\otimes\,\mathds{1}_3,
\end{equation}
where the indices indicate the qubit number. Since the driving in the Rydberg system equally affects all three qubits (atoms) and the interactions between the qubits are always present, the challenge is not only to implement the $\mathrm{CNOT}_{1,2}$ gate but also to leave the third qubit unchanged at the same time ($\mathds{1}_3$). Such problems are important, for example, to realize precise gates on a quantum computer while suppressing crosstalk, i.e., without affecting nearby qubits (see, for example, Ref.~\cite{ketterer}).

\subsubsection{Performance of the RALLY methods}
The results shown in Fig.~\ref{fig:CNOT} for the unitary infidelity for different $N_\text{L}$ and $N_\text{P}=5$ indicate that the RALLY methods become highly accurate once the number of optimization parameters surpasses the information-theoretic lower bound \(N^{2}-1 = 63\)~\cite{Control_landscape_Rabitz,Lloyd_2014}; see Appendix~\ref{app:cond}.  When compared to the dCRAB algorithm, both with Fourier and piecewise-constant (pwc) basis, the RALLY methods are more accurate while requiring fewer figure-of-merit evaluations. The RALLY methods also exhibit higher success rates for the infidelity to be below $10^{-3}$. This implies that, for practical applications where an infidelity of about $10^{-3}$ is sufficient, RALLY is advantageous in terms of the success rate and the number of figure-of-merit evaluations required.

For the above results, the Nelder–Mead optimization algorithm was used with identical settings for both RALLY and dCRAB (see Appendix~\ref{app:CNOT_coordinates} for further details). The values for the detuning $\Delta$ were limited to the interval \(\Delta\in[-10,10]\,\mathrm{MHz}\), which corresponds to realistic hardware constraints~\cite{Silverio2022pulseropensource}. To ensure that these constraints are not violated in RALLY$_\text{A}$, we enforce $\bm{\xi} \in [0,1]^{N_\text{L}}$. The total evolution time $T$ for dCRAB and RALLY$_\text{A}$ was chosen to be approximately the same as for RALLY$_\text{T}$. For the Fourier basis, the maximum bandwidth \(\Delta\Omega\) is chosen to satisfy \(\Delta\Omega\, T \geq N^{2} - 1\)~\cite{Lloyd_2014,dcrab}, in order not to limit the information content of the control pulse.

\subsubsection{Higher probability of convergence for larger \(N_\text{P}\)}

\begin{figure}[t]
  \centering
  \begin{overpic}[width=0.90\columnwidth]{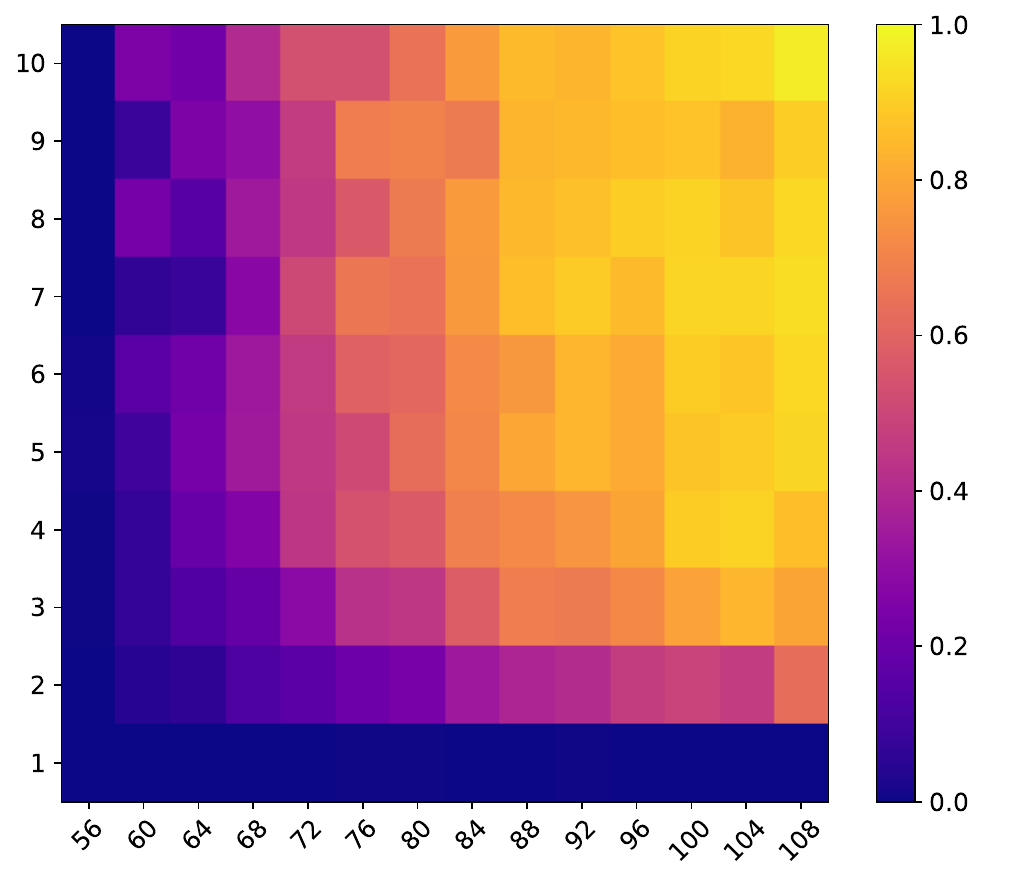}
    % Common x-axis label (centered at bottom)
    \put(50,-1){\makebox(0,0){\small Total pulse-sequence duration [$\mu$s]}}

    \put(-4,32){\rotatebox{90}{\small Layer size $N_\text{P}$}}
    \put(98,32){\rotatebox{90}{\small P$(J_\text{u}\leq 10^{-3})$}}
    % \put(-2.5,4){\rotatebox{90}{\small $\#$ FoM evals}}

  \end{overpic}
  \caption{
    Success probability for achieving \(J_\text{u}\le 10^{-3}\) for different total pulse-sequence durations and layer sizes \(N_\text{P}\) using RALLY$_\text{T}$. The results indicate that the effectiveness of the method improves as the layer size increases while keeping the total pulse-sequence duration and the number of layers fixed. For each pair of total duration and \(N_\text{P}\), we run optimizations with different random initializations of \(\tau_\ell\) and group the results by their final total pulse-sequence durations, with at least 100 runs in each group. Due to the abundance of local minima in the optimization landscape, the layer durations typically remain close to their initial values. The total number of optimization parameters is \(N_L = 66\), slightly exceeding the theoretical lower bound of 63. 
  }
  \label{fig:CNOT-heatmap}

\end{figure}

To assess the impact of a higher expressivity due to a larger layer size \(N_\text{P}\), we consider the success probability for achieving an accurate approximation of \(J_\text{u}\le 10^{-3}\). This probability is shown in Fig.~\ref{fig:CNOT-heatmap} for different values of the total pulse-sequence duration and the layer size \(N_\text{P}\). The results in Fig.~\ref{fig:CNOT-heatmap} imply that a layer size \mbox{$N_\text{P}=3$} is sufficient to achieve an accurate approximation at a relatively small total pulse-sequence duration, which corresponds to three pulses per optimization parameter (layer)~\footnote{To speed up the calculation, we used a gradient-based optimization (L-BFGS-B~\cite{LBFGSB}) but gradient-free methods yield qualitatively similar results.}. These results also confirm that the enhanced performance of the RALLY methods at $N_\text{P}>1$ is not due to longer time durations of pulse sequences, which could be caused by the higher number of pulses. We attribute the low performance of RALLY\(_\text{T}\) at relatively small $N_\text{P}$ in Fig.~\ref{fig:CNOT-heatmap} to an insufficient expressivity of the propagator \(U_{\text{R,T}}(\bm{\tau})\). Further, the low performance of RALLY\(_\text{T}\) at small total pulse-sequence durations ($\lesssim 60 \mu \text{s}$) is likely to be caused by the proximity to quantum-speed-limit bounds; see Appendix~\ref{app:cond}.

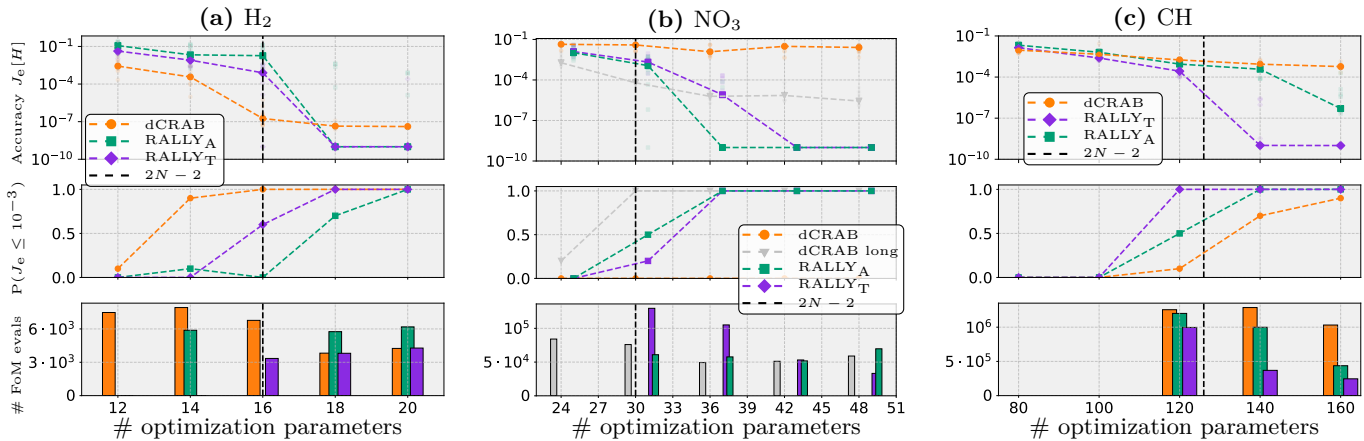
\begin{figure*}[t]
  \centering
    \begin{minipage}[t]{0.32\textwidth}
      \centering
      \begin{overpic}[width=\textwidth]{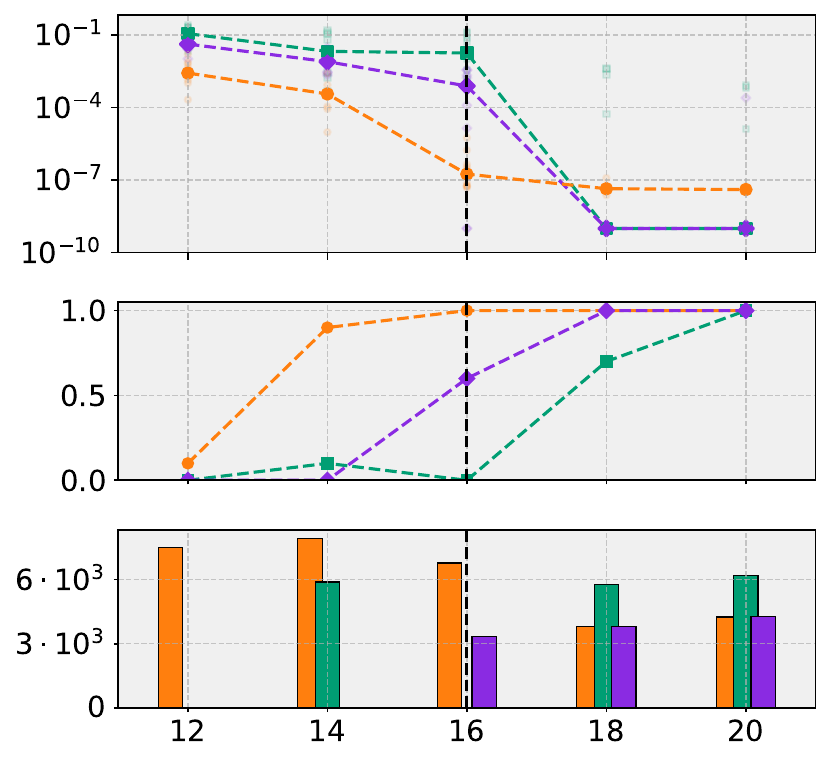}
        % Panel title
        \put(50,94){\makebox(0,0){\textbf{\small (a) $\mathrm{H}_2$}}}
    
        % Axes labels (apply to all panels)
        \put(55,-1){\makebox(0,0){\small $\#$ optimization parameters}}
        \put(-2,60){\rotatebox{90}{\tiny Accuracy $J_\text{e}[H]$}}
        % \put(-9,66){\rotatebox{90}{\small Accuracy $J_\text{e}$ [H]}}
        \put(-3,30){\rotatebox{90}{\tiny P($J_\text{e}\leq 10^{-3}$)}}
        \put(-2, 2){\rotatebox{90}{\tiny $\#$ FoM evals}}
    
        % Legend specific to H2
        \put(15,55){%
          \begin{tikzpicture}
            \node[
              fill=white,
              fill opacity=0.8,
              text opacity=1,
              draw=black,
              rounded corners=2pt,
              inner sep=1pt
            ]{%
              \begin{tikzpicture}[x=1em,y=0.7em] % compact
                % dCRAB (Fourier)
                \draw[orange, thick, dashed] (0,3) -- (1.8,3);
                \draw[orange, thick,
                      mark=*,
                      mark size=1.5pt,
                      mark options={solid, fill=orange, rotate=180}]
                  plot coordinates {(0.9,3)};
                \node[right=0.35em] at (1.8,3) {\tiny dCRAB};
    
                % RALLY_A
                \draw[myGreen, thick, dashed] (0,2) -- (1.8,2);
                \draw[myGreen, thick,
                      mark=square*,
                      mark size=1.5pt,
                      mark options={solid}]
                  plot coordinates {(0.9,2)};
                \node[right=0.35em] at (1.8,2) {\tiny RALLY$_\text{A}$};
    
                % RALLY_T
                \draw[myPurple, thick, dashed] (0,1) -- (1.8,1);
                \draw[myPurple, thick,
                      mark=square*,
                      mark size=1.5pt,
                      mark options={solid, rotate=45}]
                  plot coordinates {(0.9,1)};
                \node[right=0.35em] at (1.8,1) {\tiny RALLY$_\text{T}$};
    
                % 2N-2: black line only
                \draw[black, thick, dashed] (0,0) -- (1.8,0);
                \node[right=0.35em] at (1.8,0) {\tiny $2N-2$};
              \end{tikzpicture}%
            };
          \end{tikzpicture}%
        }
      \end{overpic}
    \end{minipage}
    \hfill
    % =============== (b) NO3 ===============
    \begin{minipage}[t]{0.325\textwidth}
      \centering
      \begin{overpic}[width=\textwidth]{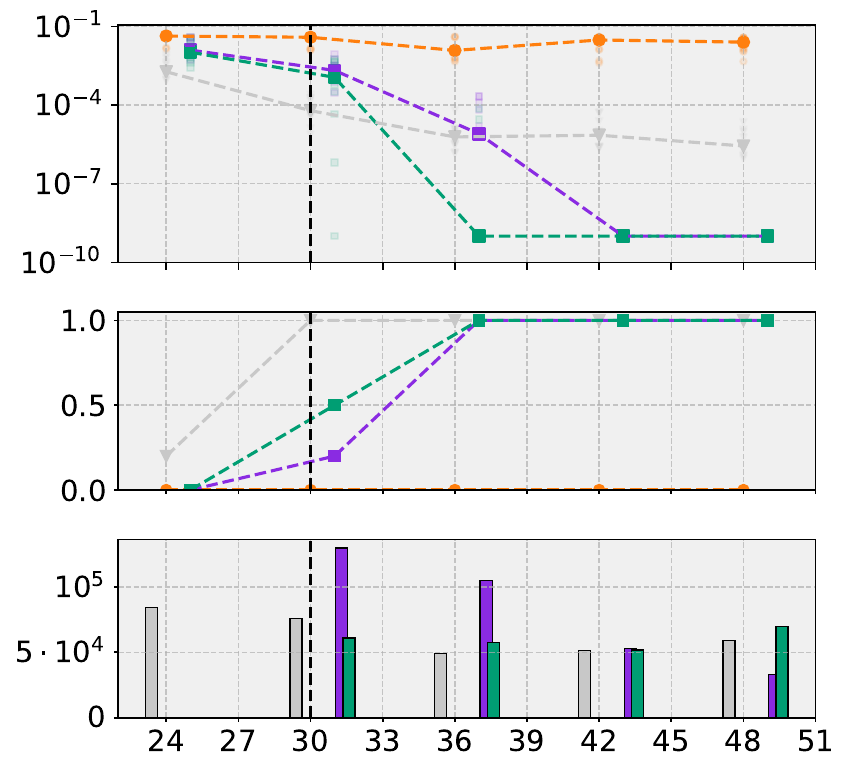}
        % Panel title
        \put(50,92){\makebox(0,0){\textbf{\small (b) $\mathrm{NO}_3$}}}
    
        % x-axis label (y-axis shared with left panel)
        \put(55,-1){\makebox(0,0){\small $\#$ optimization parameters}}
    
        % Legend specific to NO3
        \put(60,25){%
          \begin{tikzpicture}
            \node[
              fill=white,
              fill opacity=0.8,
              text opacity=1,
              draw=black,
              rounded corners=2pt,
              inner sep=1pt
            ]{%
              \begin{tikzpicture}[x=1em,y=0.7em]
                % dCRAB, T = 500 μs
                \draw[orange, thick, dashed] (0,4) -- (1.8,4);
                \draw[orange, thick,
                      mark=*,
                      mark size=1.5pt,
                      mark options={solid, fill=orange, rotate=180}]
                  plot coordinates {(0.9,4)};
                \node[right=0.35em] at (1.8,4) {\tiny dCRAB};
    
                % dCRAB, T = 3000 μs (lighter gray)
                \draw[lightgray, thick, dashed] (0,3) -- (1.8,3);
                \draw[lightgray, thick,
                      mark=triangle*,
                      mark size=1.5pt,
                      mark options={solid, rotate=180}]
                  plot coordinates {(0.9,3)};
                \node[right=0.35em] at (1.8,3) {\tiny dCRAB long};
    
                % RALLY_A
                \draw[myGreen, thick, dashed] (0,2) -- (1.8,2);
                \draw[myGreen, thick,
                      mark=square*,
                      mark size=1.5pt,
                      mark options={solid}]
                  plot coordinates {(0.9,2)};
                \node[right=0.35em] at (1.8,2) {\tiny RALLY$_\text{A}$};
    
                % RALLY_T
                \draw[myPurple, thick, dashed] (0,1) -- (1.8,1);
                \draw[myPurple, thick,
                      mark=square*,
                      mark size=1.5pt,
                      mark options={solid, rotate=45}]
                  plot coordinates {(0.9,1)};
                \node[right=0.35em] at (1.8,1) {\tiny RALLY$_\text{T}$};
    
                % 2N-2: black line only
                \draw[black, thick, dashed] (0,0) -- (1.8,0);
                \node[right=0.35em] at (1.8,0) {\tiny $2N-2$};
              \end{tikzpicture}%
            };
          \end{tikzpicture}%
        }
      \end{overpic}
    \end{minipage}
    \hfill
    % =============== (c) CH ===============
    \begin{minipage}[t]{0.32\textwidth}
      \centering
      \begin{overpic}[width=\textwidth]{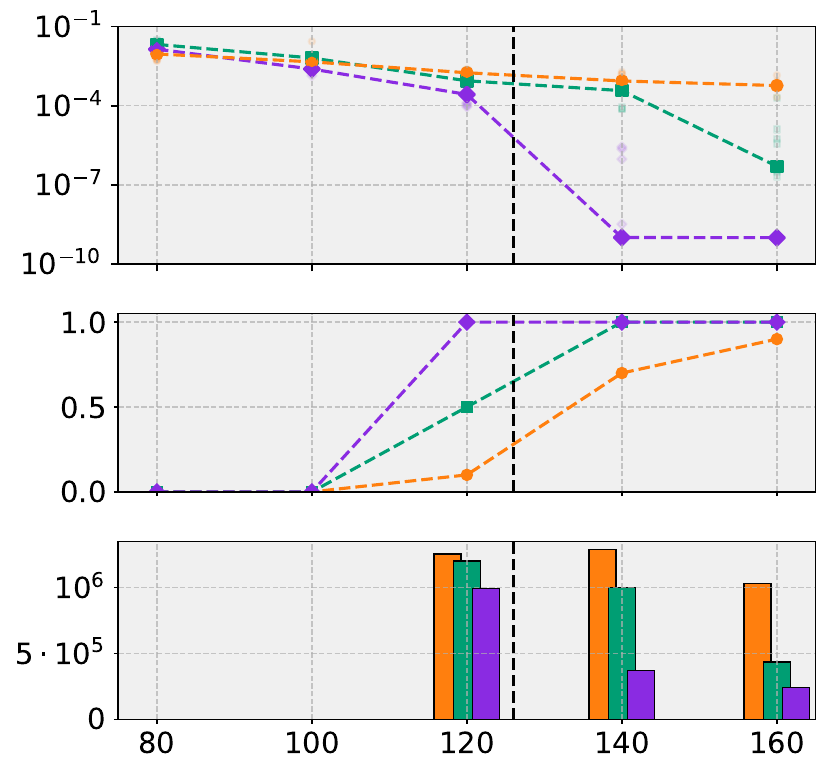}
        % Panel title
        \put(50,94){\makebox(0,0){\textbf{\small (c) $\mathrm{CH}$}}}
    
        % x-axis label
        \put(55,-1){\makebox(0,0){\small $\#$ optimization parameters}}
    
        % Legend specific to CH
        \put(20,60){%
          \begin{tikzpicture}
            \node[
              fill=white,
              fill opacity=0.8,
              text opacity=1,
              draw=black,
              rounded corners=2pt,
              inner sep=1pt
            ]{%
              \begin{tikzpicture}[x=1em,y=0.7em] % compact
                % dCRAB (Fourier)
                \draw[orange, thick, dashed] (0,3) -- (1.8,3);
                \draw[orange, thick,
                      mark=*,
                      mark size=1.5pt,
                      mark options={solid, fill=orange, rotate=180}]
                  plot coordinates {(0.9,3)};
                \node[right=0.35em] at (1.8,3) {\tiny dCRAB};
    
                % RALLY_T
                \draw[myPurple, thick, dashed] (0,2) -- (1.8,2);
                \draw[myPurple, thick,
                      mark=square*,
                      mark size=1.5pt,
                      mark options={solid, rotate=45}]
                  plot coordinates {(0.9,2)};
                \node[right=0.35em] at (1.8,2) {\tiny RALLY$_\text{T}$};
    
                % RALLY_A
                \draw[myGreen, thick, dashed] (0,1) -- (1.8,1);
                \draw[myGreen, thick,
                      mark=square*,
                      mark size=1.5pt,
                      mark options={solid}]
                  plot coordinates {(0.9,1)};
                \node[right=0.35em] at (1.8,1) {\tiny RALLY$_\text{A}$};

                \draw[black, thick, dashed] (0,0) -- (1.8,0);
                \node[right=0.35em] at (1.8,0) {\tiny $2N-2$};
              \end{tikzpicture}%
            };
            
                % 2N-2: black line only
            %     \draw[black, thick, dashed] (0,0) -- (1.8,0);
            %     \node[right=0.35em] at (1.8,0) {\tiny $2N-2$};
            %   \end{tikzpicture}%
            % };
          \end{tikzpicture}%
        }
      \end{overpic}
    \end{minipage}
\caption{
  Performance of the RALLY methods and dCRAB for the ground-state preparation for the three molecules \(\mathrm{H}_2\)~(a), \(\mathrm{NO}_3\)~(b), and \(\mathrm{CH}\)~(c). The RALLY methods (with \(N_\text{P}=5\)) achieve the target precision of $J_\text{e}\leq10^{-9}$\,H for all three molecules close to the information-theoretic bound (dashed vertical line) on the minimum number of parameters (see Appendix~\ref{app:cond}). Except for \(\mathrm{H}_2\), the RALLY methods converge to the target accuracy faster than dCRAB and require fewer figure-of-merit evaluations. Each panel shows as a function of the number of optimization parameters (\(N_L\) for RALLY): median ground-state energy accuracy $J_\text{e}$~\eqref{eq:gs_energy} (top), success probability for achieving $J_\text{e}\leq10^{-3}$\,H over 10 optimization runs (middle), and the median number of figure-of-merit evaluations required to reach $J_\text{e}\leq10^{-3}$\,H (bottom). Missing bars in the bottom panels correspond to cases when no successful optimizations were achieved. All methods use the same optimizer (adaptive Nelder–Mead) and identical stopping criteria. For \(\mathrm{NO}_3\), two dCRAB settings are shown: one with total time $T=500\,\mu$s similar to RALLY’s optimized schedule and one with a much longer $T=3000\,\mu$s illustrating the need for dCRAB to operate at longer times for this problem. For dCRAB, the Fourier basis was used.
  }
  \label{fig:VQS_dcrab_comp}
\end{figure*}

\subsubsection{Reducing the total pulse-sequences time duration with an additional penalty term} 
To reduce the total pulse-sequence time duration $\sum_{\ell=1}^{N_L}\tau_{\ell}$, RALLY$_\text{T}$ offers the possibility to add a penalty term $(\sum_{\ell=1}^{N_L}\tau_{\ell})^2$ to the figure of merit, which is to be minimized. For the unitary-synthesis problem, we hence obtain $J_\text{u}[u_k] + \lambda \left(\sum_{\ell=1}^{N_L}\tau_{\ell}\right)^2$ for the total figure of merit, where $\lambda$ is the weight of the additional time-reducing term. Thus, a single optimization run with a suitable weight $\lambda$ can efficiently yield a short pulse sequence and an accurate approximation at the same time.

For the unitary-synthesis problem, we obtain a short pulse sequence with a total time duration of about $55 \mu \text{s}$ that achieves a unitary infidelity $J_\text{u}\leq 10^{-3}$ using $\lambda=10^{-4}$. This time duration is below the bound provided by the standard quantum speed limit, which ignores the geometry of the drift and control Hamiltonians~\cite{qsl_unitary_2}. Hence, a more appropriate reference value for the globally-driven Rydberg platform may be the geometric speed limit~\cite{GSL1,qsl_spinchain}.

\subsection{Ground-state preparation for molecules} 
For the variational ground-state preparation, we use the energy-expectation value of the target Hamiltonian $\cal{H}_\text{T}$ as the figure of merit to be minimized
\begin{equation} \label{eq:gs_energy}
    J_\text{e}[u]\equiv\langle\psi_0|U^\dagger[u]{\mathcal{H}}_\text{T}U[u]|\psi_0\rangle,
\end{equation}
where $|\psi_0\rangle=\ket{0}^{\otimes n}$ is the initial state of the system. Similar problems are considered in variational quantum algorithms such as the variational quantum eigensolver and the variational quantum simulation~\cite{tilly_variational_2022,Kokail_2019,ctrl-VQE,montangero_qsl_2,qoc_vqa}.

Using the RALLY methods, we determine the ground states of the molecules H$_2$, NO$_3$ and CH in numerical simulations of a globally-driven Rydberg-atom platform~\eqref{eq:analog_sim}. The minimal meaningful choice of the spin orbitals for the three molecules H$_2$, NO$_3$ and CH results in an increasing Hilbert-space (or symmetry-sector) dimension of 9, 16 and 64, respectively, as detailed in Appendix~\ref{app:VQS_optimization}. The three molecules also differ in the von-Neumann entanglement $S$ of the corresponding ground states with the average single-qubit (spin-orbital) entanglement: $S\approx0.1$ for H$_2$, $S\approx0.8$ for NO$_3$ and $S\approx0.4$ for CH; see Appendix~\ref{app:entropy}.

The RALLY methods achieve the target precision of $10^{-9}$\,H while approaching the information-theoretic bound for all three molecules as shown in Fig.~\ref{fig:VQS_dcrab_comp}. The RALLY methods also yield more accurate results than dCRAB for the molecules, while requiring fewer figure-of-merit evaluations.  While the accuracy of the dCRAB algorithm gradually improves when increasing the number of optimization parameters, RALLY methods converge to the target precision once $N_\text{L}$ is above the information-theoretic minimum. Only for the H$_2$ molecule with the smallest symmetry-sector dimension and lowest ground-state entanglement, dCRAB outperforms the RALLY methods concerning the accuracy at smaller numbers of the optimization parameters and the success probability for achieving  $J_\text{e}\leq 10^{-3}$. For this comparison, identical optimization algorithms, optimization settings, and convergence criteria were used, as detailed in Appendix~\ref{app:VQS_optimization}. It is also noteworthy that the simulations using the RALLY$_\text{T}$ method terminated two to five times faster than those using the dCRAB and RALLY$_\text{A}$ algorithms.

For the $\mathrm{NO}_3$ molecule, the dCRAB algorithm fails to reach an accuracy of $10^{-3}\,\text{H}$ within a total pulse-sequence duration of $T = 500\,\mu\text{s}$, which is comparable to the durations obtained with RALLY$_\text{T}$ and equal to that used for RALLY$_\text{A}$, cf Fig.~\ref{fig:VQS_dcrab_comp}b. Accuracy below $10^{-3}\,\text{H}$ was reached for dCRAB when increasing the total control-field time to $T=3000\,\mu$s.

\subsection{State transfer in an Ising spin chain} \label{sec_ising}
Let us now consider an Ising spin chain with \(n\) \mbox{spins-1/2} described by the Hamiltonian~\eqref{eq:random_spin_H}. We consider a state-transfer problem, when the initial state \(\ket{\psi_0}=\ket{0}^{\otimes n}\) must evolve into the \(n\)-spin GHZ state $\ket{\mathrm{GHZ}_n}=\frac{1}{\sqrt{2}}\bigl(\ket{0}^{\otimes n}+\ket{1}^{\otimes n}\bigr)$. The figure of merit used is the state-transfer infidelity
\begin{equation}  \label{eqn_state-trans}
    J_\text{s}[u]\equiv1-\bigl|\bra{\mathrm{GHZ}_n}U[u]|\psi_{0}\rangle\bigr|^{2}.
\end{equation}

\begin{figure}[b]
  \centering
  \begin{overpic}[width=0.8\columnwidth]{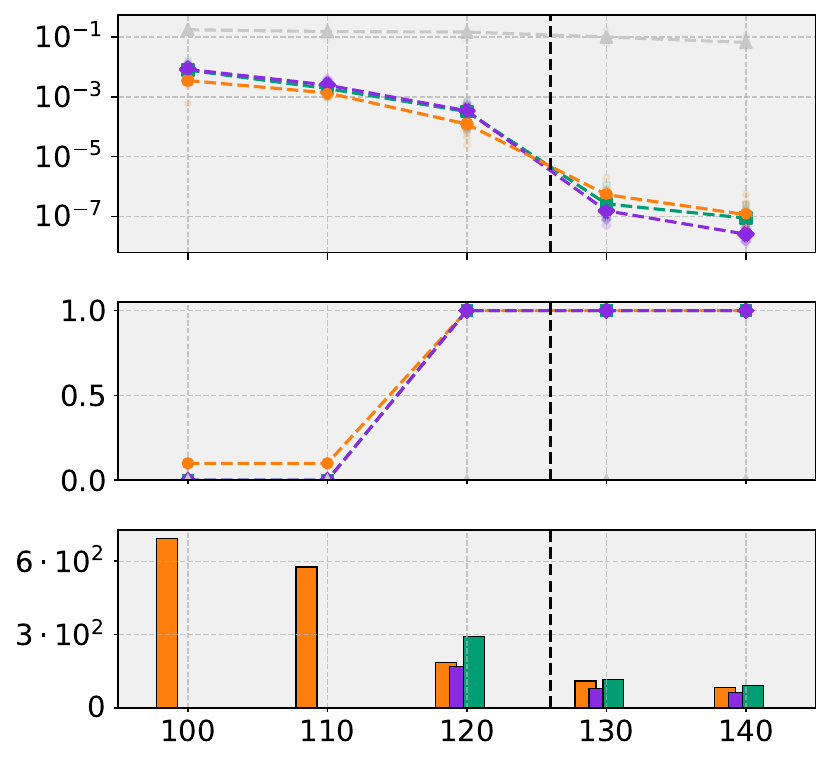}
    % Common x-axis label (centered at bottom)
    \put(55,-1){\makebox(0,0){\small $\#$ optimization parameters}}

    \put(-2.5,64){\rotatebox{90}{\small Infidelity $J_\text{s}$}}
    \put(-2.9,32){\rotatebox{90}{\small P$(J_\text{s}\leq 10^{-3})$}}
    \put(-2.5,2){\rotatebox{90}{\small $\#$ FoM evals}}

    \put(50,20){%
      \begin{tikzpicture}
        \node[
          fill=white,
          fill opacity=0.8,
          text opacity=1,
          draw=black,
          rounded corners=2pt,
          inner sep=2pt
        ] {
          \begin{tikzpicture}[x=1em,y=1em]
            % dCRAB Fourier: orange, line + triangle down
            \draw[orange, thick, dashed] (0,4) -- (1.8,4);
            \draw[orange, thick,
                  mark=triangle*,
                  mark size=1.5pt, 
                  mark options={solid, fill=orange, rotate=180}]
              plot coordinates {(0.9,4)};
            \node[right=0.6em] at (1.8,4) {\scriptsize GRAPE};

            % RALLY_A: green, line + square
            \draw[myGreen, thick, dashed] (0,3) -- (1.8,3);
            \draw[myGreen, thick,
                  mark=square*,
                  mark size=1.5pt, 
                  mark options={solid}]
              plot coordinates {(0.9,3)};
            \node[right=0.6em] at (1.8,3) {\scriptsize RALLY$_\text{A}$};
    
            % RALLY_T: purple, line + diamond
            \draw[myPurple, thick, dashed] (0,2) -- (1.8,2);
            \draw[myPurple, thick,
                  mark=square*,
                  mark size=1.5pt, 
                  mark options={solid, rotate=45}]
              plot coordinates {(0.9,2)};
            \node[right=0.6em] at (1.8,2) {\scriptsize RALLY$_\text{T}$ with \(N_\text{P}=5\)};

            % RALLY_T: purple, line + diamond
            \draw[lightgray, thick, dashed] (0,1) -- (1.8,1);
            \draw[lightgray, thick,
                  mark=triangle*,
                  mark size=1.5pt, 
                  mark options={solid, rotate=0}]
              plot coordinates {(0.9,1)};
            \node[right=0.6em] at (1.8,1) {\scriptsize RALLY$_\text{T}$ with \(N_\text{P}=1\)};
    
            % N^2 - 1: black line only
            \draw[black, thick, dashed] (0,0) -- (1.8,0);
            \node[right=0.6em] at (1.8,0) {\scriptsize $2N-2$};
          \end{tikzpicture}
        };
      \end{tikzpicture}%
    }

  \end{overpic}
    \caption{
    State transfer in the Ising spin chain of $n=6$ spins-1/2: RALLY methods (with \mbox{\(N_\text{P}=5\)}) and GRAPE exhibit similar results, whereas RALLY\(_\text{T}\) better reflects experimental constraints, namely bandwidth-limited pulses (cf. Sec.~\ref{sec_rally}) with control amplitudes restricted to two discrete values $\{+1,-1\}$. The individual parts of the figure show as a function of the number of optimization parameters ($N_\text{L}$ for RALLY): median state infidelity $J_\text{s}$~\eqref{eqn_state-trans} (top), success probability for achieving $J_\text{s}\leq 10^{-3}$ over 10 optimization runs (middle), and the median number of figure-of-merit evaluations required to achieve $J_\text{s}\leq 10^{-3}$ (bottom). Missing bars in the bottom panel correspond to cases when no successful optimizations were achieved. All methods use the same gradient-based optimizer (L-BFGS-B) and identical stopping criteria for the optimization.
    }
    \label{fig:grape_comp}

\end{figure}
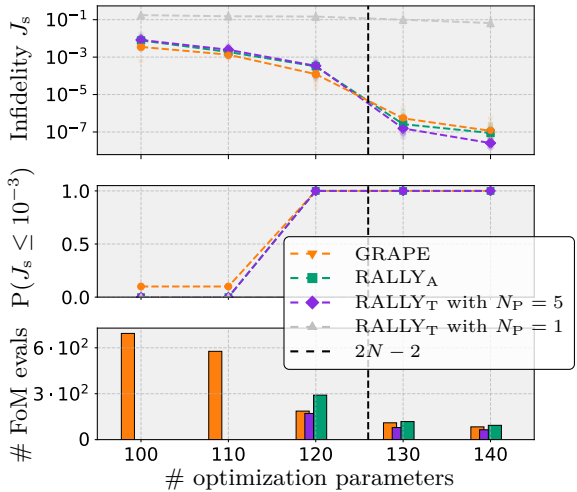

\subsubsection{Performance of the RALLY methods} \label{subsub:RALLY_GRAPE}
To demonstrate the capabilities of RALLY$_\text{T}$, we (i) restrict the pulse amplitudes to two discrete values \mbox{$J(t)\in\{+1,-1\}$} and (ii) smoothly interpolate between these two discrete pulse amplitudes, cf. Fig.~\ref{fig:bandwidth}. The latter is to accommodate a finite bandwidth of the experimental setting. As shown in Fig.~\ref{fig:grape_comp}, RALLY$_\text{T}$ achieves high-fidelity state transfer on an $n=6$ Ising spin chain, where the sigmoid function used for the interpolation is detailed in Appendix~\ref{app:sigmoid}. Once the number of optimization parameters $N_\text{L}$ exceeds the information-theoretic minimum (see Appendix~\ref{app:cond}), RALLY$_\text{T}$ reaches an infidelity of $10^{-7}$ and lower. Enforcing the above hardware constraints has little effect on these results: RALLY$_\text{T}$ exhibits comparable accuracy when choosing the pulse amplitudes randomly from a continuous distribution and excluding the interpolation between pulses. GRAPE and RALLY$_\text{A}$ achieve similar infidelities for the state-transfer problem but without accommodating the constraints on the finite bandwidth and discrete control-amplitude values. In fact, the above hardware constraints are difficult to integrate into GRAPE or RALLY$_\text{A}$ without significant computational overhead. 

It is remarkable that GRAPE and RALLY$_\text{A}$ achieve similar infidelities because GRAPE is equivalent to RALLY$_\text{A}$ with $N_\text{P}=1$. This implies that, when optimizing the pulse amplitudes, small layer sizes $N_\text{P}$ are already sufficient to achieve very small infidelities for the given state-transfer problem. In contrast, RALLY$_\text{T}$ requires $N_\text{P}>1$ for achieving small infidelities as shown in Fig.~\ref{fig:grape_comp}. In general, the performance of optimal-control methods is problem-specific, but the introduction of random layers provides a means to achieve accurate approximations for a broad range of quantum-optimal-control problems by choosing a suitable layer size. 

Since analytic expressions for the gradient of the state-transfer infidelity in Eq.~\eqref{eqn_state-trans} can be incorporated in the RALLY methods and GRAPE (see Appendix~\ref{app:grad}), the above numerical simulations were performed with the gradient-based optimizer L-BFGS-B with the same settings for all methods. As generally expected from gradient-based optimizations, the number of figure-of-merit evaluations is orders of magnitude smaller than for derivative-free optimizations described above.

\subsubsection{Scaling of the optimization runtime with the system size}
Let us now compare the total runtime required for the pulse-sequence optimizations of the state-transfer problem for GRAPE and RALLY$_\text{T}$, where we also include any pre-processing routines.
\begin{figure}[t]
  \centering
  \begin{overpic}[width=0.8\columnwidth]{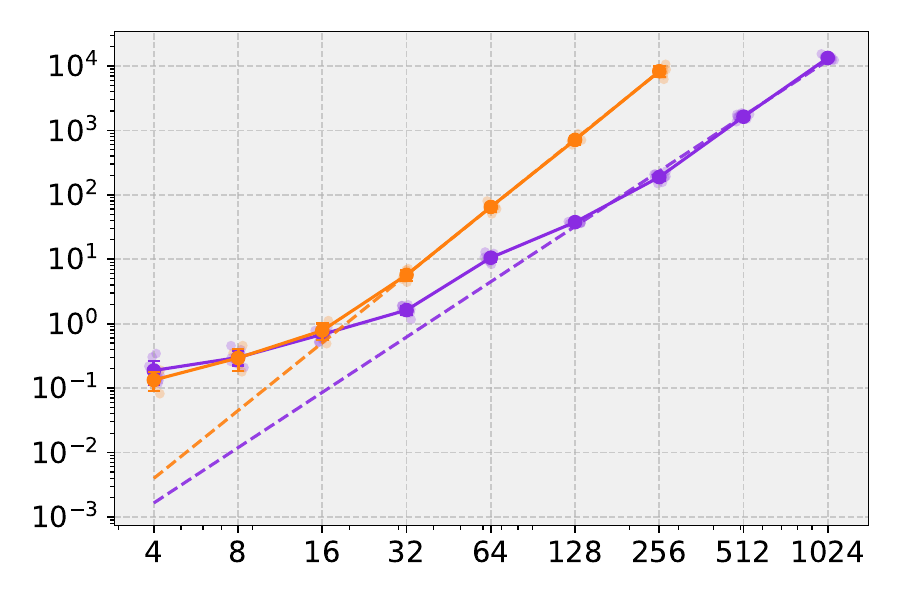}
    % Common x-axis label (centered at bottom)
    \put(55,-1){\makebox(0,0){\small Hilbert-space dimension $N$}}

    \put(-4,20){\rotatebox{90}{\small Total runtime\,[s]}}
    % \put(-2.5,4){\rotatebox{90}{\small $\#$ FoM evals}}

    \put(45,15){%
      \begin{tikzpicture}
        \node[
          fill=white,
          fill opacity=0.8,
          text opacity=1,
          draw=black,
          rounded corners=2pt,
          inner sep=2pt
        ] {
          \begin{tikzpicture}[x=1em,y=1em]
            % dCRAB Fourier: orange, line + triangle down
            \draw[orange, thick] (0,4) -- (1.8,4);
            \draw[orange, thick,
                  mark=*,
                  mark size=1.5pt, 
                  mark options={solid, fill=orange, rotate=180}]
              plot coordinates {(0.9,4)};
            \node[right=0.6em] at (1.8,4) {\scriptsize GRAPE $\mathcal{O}(N^{2.85})$};
    
            % RALLY_T: purple, line + diamond
            \draw[myPurple, thick] (0,3) -- (1.8,3);
            \draw[myPurple, thick,
                  mark=*,
                  mark size=1.5pt, 
                  mark options={solid, rotate=45}]
              plot coordinates {(0.9,3)};
            \node[right=0.6em] at (1.8,3) {\scriptsize RALLY$_\text{T}$ $\mathcal{O}(N^{3.52})$};
    
          \end{tikzpicture}
        };
      \end{tikzpicture}%
    }

  \end{overpic}
  \caption{
    Runtime scaling to achieve a state-transfer infidelity $J_\text{s}\leq10^{-3}$ for different Hilbert-space dimensions \(N = 2^n\) for the Ising spin chain: RALLY$_\text{T}$ (violet) and GRAPE (orange). Fitting a linear function to the last four points in the double-logarithmic plot yields for the power-law functions the exponents \(2.85 \pm 0.13\) for RALLY$_\text{T}$ and \(3.52 \pm 0.02\) for GRAPE, indicating a better scaling for RALLY$_\text{T}$. Within the maximal runtime of \(2\times10^4\)\,s, RALLY$_\text{T}$ is able to handle systems with two spins more than GRAPE. Each point is an average over 10 independent calculations, where the same random coefficient sets \(h_{ix},h_{iz}\in[0.5,1]\) were used for GRAPE and RALLY$_\text{T}$.
  }
  \label{fig:time_grape_comp}

\end{figure}
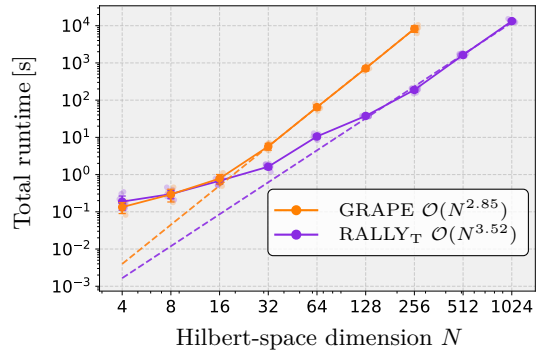
Results for the runtime scaling shown in Fig.~\ref{fig:time_grape_comp} indicate that, within a maximal runtime of \(2\cdot10^{4}\,\mathrm{s}\), RALLY$_\text{T}$ can handle systems with two spins more than GRAPE. Fitting the results in a double-logarithmic plot by a linear function yields for the runtime scaling the exponents \(2.85 \pm 0.13\) for RALLY$_\text{T}$ and \(3.53 \pm 0.02\) for GRAPE. On the basis of theoretical estimates, we expect the scaling to be $\mathcal{O}(N^{3,5})$ for RALLY$_\text{T}$ and $\mathcal{O}(N^{4,5})$ for GRAPE, see Appendix~\ref{app_scaling}. We attribute this discrepancy to the fact that the theoretical estimates correspond to worst-case estimates and that the scaling of the individual numerical operations has not yet reached its asymptotic regime for the system sizes considered here.

Considering system sizes beyond $n=10$ spins is feasible using exact diagonalization of the Hamiltonian but the RALLY methods can also be combined with other numerical techniques for the time evolution of quantum systems, such as the Suzuki-Trotter expansion~\cite{suzuki} or tensor-network techniques~\cite{SCHOLLWOCK, CRAB}.

\section{Conclusion} \label{sec_conclude}
We introduced a family of efficient parametrized pulse sequences based on random layers for quantum optimal control. The efficacy of the random-layers (RALLY) methods is based on two construction principles: \mbox{(1) \textit{Random unitaries}} - The ensemble of reachable unitaries converges exponentially fast to the uniform Haar-random ensemble with the total number of pulses $N_\text{P} N_\text{L}$, such that the unitary space can be explored efficiently. (2) \textit{Layers} - Grouping pulses into layers, with one optimization parameter per layer, reduces the number of optimization parameters and simplifies the optimization. Having these properties, the RALLY pulse sequences solve the long-standing problem in quantum optimal control of finding an optimal pulse structure that leads to an efficient exploration of the unitary space with a minimal number of optimization parameters. 

We focused on two outstanding RALLY methods: RALLY$_\text{T}$ optimizes the time durations of the layers keeping the (potentially discrete) random amplitudes fixed, while RALLY$_\text{A}$ optimizes a factor that scales the randomly chosen amplitudes keeping the pulse durations fixed. We found in numerical simulations that the RALLY methods approach an information-theoretic lower bound on the number of optimization parameters~\cite{Lloyd_2014}. In the gradient-free setting important for closed-loop optimizations, RALLY methods are orders of magnitude more accurate, while also requiring fewer figure-of-merit evaluations when compared to existing algorithms. Furthermore, RALLY$_\text{T}$ can be simulated efficiently, while also accounting for limited-bandwidth constraints. By incorporating an explicit penalty on the total time, RALLY$_\text{T}$ provides an efficient route for finding short pulses in large systems. Both RALLY methods are suitable for both closed-loop and open-loop optimizations. 

The RALLY methods could facilitate designing efficient quantum‑inspired classical optimal‑control methods. The application of the random-layer construction concepts to pulse sequences in other bases, such as the Fourier basis, is a further interesting research question. Beyond quantum optimal control, the random-layers structure can be directly embedded into parameterized circuits for variational quantum algorithms and quantum machine learning. For variational quantum algorithms, the RALLY concept can be used for enhanced ansätze, for example for the variational quantum eigensolver~\cite{qoc_vqa}. For quantum machine learning models like quantum neural networks, RALLY can provide more expressive feature maps or trainable models that avoid over-parameterization. The possibility of finding short pulses with RALLY$_\text{T}$ can facilitate the study of quantum speed limits.

The source code is published in a GitHub repository~\cite{code}.

\begin{acknowledgments}
M.D. and W.H. thank R. Zeier and V. V. Dobrovitski for very helpful discussions. M.D. and W.H. acknowledge financial support from the project HPCQS. HPCQS has received funding from the European High-Performance Computing Joint Undertaking (JU) under grant agreement No 101018180. The JU receives support from the European Union’s Horizon 2020 research and innovation programme and Germany, France, Italy, Ireland, Austria and Spain in equal parts. The author is solely responsible for its content, it does not represent the opinion of the EuropeanHPC Joint Undertaking and the EuroHPC JU is not responsible for any use that may be made of the information it contains. M.K. acknowledges funding from the European Union’s Horizon Europe research and innovation program under grant agreement No 101135699.
\end{acknowledgments}

\appendix

\section{Necessary conditions for solving quantum-optimal-control problems} \label{app:cond}

Finding accurate solutions to a quantum control problem requires meeting the following complementary requirements.

\subsection{Reachability and controllability}\label{app_reachability}
Reachability concerns the question of whether, given a set of control fields, a specific target unitary transformation is reachable in principle. A stronger requirement going beyond reachability is full controllability, which corresponds to the ability to generate any  unitary transformation. The common test for controllability is to investigate whether the dynamical Lie algebra spans the traceless skew-Hermitian algebra $\mathfrak{su}(N)$~\cite{controllability1}. It is generally more difficult to establish reachability than full controllability~\cite{gs-reachability,burgarth}.

Let us discuss the problem of the reachability of a target state $|\psi_\text{T}\rangle$ when starting from an initial state $|\psi_0\rangle$. System's dynamical Lie algebra (DLA) $\mathfrak{g}$ is defined for the control Hamiltonian ${\cal H}_c$ and the drift Hamiltonian ${\cal H}_0$ by~\cite{controllability1}
\begin{align}
\mathfrak{g} &\equiv \mathrm{Lie}(\{{i \cal H}_0,  {i \cal H}_c\})\\
&= \mathrm{span}\left\{i {\cal H}_k , [{i\cal H}_k, {i \cal H}_l], [{i \cal H}_k, [{i \cal H}_l, {i \cal H}_m]], \ldots \right\},
\end{align}
where $\mathrm{Lie}(\cdot)$ indicates the Lie closure, $k, l, m \in \{0,c\}$, and the span is taken over all possible nested commutators. If the DLA spans the entire skew-Hermitian Lie algebra $\mathfrak{u}(N)$ (or the traceless skew-Hermitian algebra $\mathfrak{su}(N)$) for an $N$-dimensional Hilbert space, the exponential map generating the corresponding dynamical Lie group yields the full unitary group $\mathbf{U}(N)$ (or the special unitary group $\mathbf{SU}(N)$). This would imply full controllability of the system~\cite{controllability1,koch_contr}.

In the presence of symmetries that induce invariant subspaces on the Hilbert space, often referred to as symmetry sectors, the total Hamiltonian can be block-diagonalized. The DLA can in this case be projected onto these sectors to verify their individual controllability. Crucially, if the initial state $|\psi_0\rangle$ and the target state $|\psi_{\text{T}}\rangle$ completely reside within the same symmetry sector and the subspace is controllable, the reachability of the target state from the initial state is guaranteed~\cite{gs-reachability}.

\subsection{Quantum speed limits}
The minimal evolution time required to approximate a target unitary transformation is generally bounded from below by constraints from the Hamiltonian and the initial conditions~\cite{montangero_qsl,montangero_qsl_2,wilhelm_qsl,caneva_qsl}. Such lower bounds are formalized by different quantum speed limits depending on the problem~\cite{MT_bound1,MT_bound2, MT_bound3,qsl,qsl_unitary_2}. For example, for arbitrary state transfer in a spin chain with $n$ spins, access to $n$ local control Hamiltonians yields a minimal evolution time that scales linearly with $n$~\cite{qsl_spinchain}. To incorporate the geometry induced by the drift and control Hamiltonians, the geodesic length defined by the quantum geometric tensor can be considered~\cite{GSL1,qsl_unitary_1}.

\subsection{Information–theoretic constraint on the number of optimization parameters}
The minimal number of real optimization parameters for an arbitrarily accurate approximation corresponds to the number of real degrees of freedom required for the description of the optimization problem~\cite{Control_landscape_Rabitz,Lloyd_2014}. For example, an arbitrary quantum state in a $N$-dimensional Hilbert space is described by $2N-2$ real parameters, where the normalization condition and the arbitrariness of the global phase result in the reduction by 2. Hence, at least $2N-2$ real optimization parameters for an accurate approximation of any target state are required~\cite{Control_landscape_Rabitz,Params_manybody,dcrab}. Similarly, when approximating an arbitrary unitary from the special unitary group $\mathbf{SU}(N)$, $N^{2}-1$ optimization parameters are required. These constraints provide a useful estimate for the required number of optimization parameters. Over-parametrization is usually considered relative to these constraints.

\section{dCRAB and GRAPE algorithms} \label{app:dCRAB_GRAPE}
Quantum-optimal-control algorithms can be gradient-free and gradient-based~\cite{qoc_review_2015,koch-review}. Gradient-free methods are simple to deploy and useful for closed-loop optimization because they rely only on figure-of-merit evaluations. However, they typically converge slower than gradient-driven approaches, especially in high-dimensional control spaces. An example is the dCRAB algorithm combined with a Nelder–Mead optimizer~\cite{dcrab, dCRAB_bases,Muller_2022,nelder, adaptie_Nelder}. Gradient-based techniques converge faster and often scale better for larger systems but obtaining reliable gradients and maintaining differentiability make their direct use in closed-loop calibration difficult. The GRAPE algorithm is a canonical example~\cite{GRAPE,GRAPE_feedback,goat}. We detail the dCRAB and GRAPE algorithms below and use them for benchmarking the RALLY methods.

\paragraph{dCRAB:} The dressed chopped random basis algorithm~\cite{dcrab} is based on an expansion of the control field \(u(t)=\sum_i c_i\,\phi_i(t)\) in a truncated basis (e.g. the Fourier basis)~\cite{dCRAB_bases}, where the coefficients \(c_i\) are the optimization parameters. One key challenge is the emergence of local traps: spurious local minima induced by truncating the infinite-dimensional space to a finite basis~\cite{dcrab}. To avoid false traps, the dCRAB algorithm iteratively introduces new random basis functions \(\{\phi^{(j)}_i(t)\}\) in so-called super-iterations and reuses the optimal solution \(\tilde{u}^{(j-1)}(t)\) from the optimization in the previous super-iteration
\begin{equation} \label{eq:dcrab_pulses}
u^{(j)}(t)=c^{(j)}_{0}\,\tilde{u}^{(j-1)}(t)
            +\sum_i c^{(j)}_{i}\,\phi^{(j)}_{i}(t),
\end{equation}
where only the new coefficients \(c^{(j)}_{i}\) are optimized~\cite{Muller_2022}. The dCRAB algorithm can easily incorporate experimental constraints and it often saturates the information-theoretic lower bounds~\cite{CRAB, sensing_neu}.

\paragraph{GRAPE:} In the gradient-ascent pulse engineering algorithm~\cite{GRAPE}, the total evolution time \(T\) is divided into \(M\) equal time windows of duration \(\Delta t = T/M\). The control field \(u[j]\) is chosen to be piecewise constant for each time window \(j=1,\dots,M\). The Hamiltonian in time window \(j\) reads \({\cal H}_j \equiv {\cal H}_0 + u[j]\,{\cal H}_c\) with the corresponding propagator \(U_j \equiv \exp\!\left(-i\,\Delta t\, {\cal H}_j\right)\) and total time-evolution operator \(U(T) = U_M U_{M-1}\cdots U_1\). Closed-form expressions for the derivatives can be derived with respect to the optimization parameters \(\partial_{u[j]} U(T)\) (for more details, see Appendix~\ref{app:grad}). These analytic gradients enable efficient use of first- and quasi-second-order optimizers (for example, %conjugate gradient, 
L-BFGS-B%~\cite{CG,LBFGSB}) 
~\cite{LBFGSB}) and facilitate the consideration of constraints on the controls.

\section{Proof of convergence of the RALLY propagator $U_{\mathrm R}$}
\label{app:t-unitary convergence}

\newcommand{\U}{\mathrm{U}}

The goal of this section is to prove the theorem stated in Sec.~\ref{sec_rally} on the convergence properties of the RALLY propagator $U_{\mathrm R}$. In the following, we distinguish between two cases:
\begin{enumerate}
  \item[a.] control amplitudes drawn from a \textit{continuous} probability
        distribution, and
  \item[b.] control amplitudes drawn from a \textit{discrete} probability
        distribution (e.g., a finite set of values).
\end{enumerate}
In the continuous case (a), we show below the exponential convergence of the ensemble of unitaries reachable by $U_{\mathrm R}$ to the Haar measure with $N_\text{P} N_\text{L}$~\cite{Haar_convergence,Random_pulses_Haar}. In the discrete case (b), we use a weaker notion of convergence: we limit the proof of convergence to polynomial test functions in the matrix entries (and their complex conjugates) of $U_{\mathrm R}$. In this setting, the results of Ref.~\cite{Haar_convergence} still imply exponential convergence in $N_\text{P} N_\text{L}$ of the expectation values of all such polynomials. This is sufficient for our practical purposes, since figures of merit in quantum optimal control are usually polynomial functions of $U_{\mathrm R}$.

Consider the general form of the RALLY propagator $U_{\mathrm R}=\prod_{\ell=1}^{N_\text{L}} U_{\ell}$,
where each $U_\ell$ is a layer composed of $N_\text{P}$ pulses $U_\ell = \prod_{p=1}^{N_\text{P}} U_{\ell,p}$. For the proof, we focus on RALLY$_T$ in the following. However, the proof is readily extended to RALLY$_A$. We divide the proof into three steps: first, we establish key properties of a single layer; second, we lift these properties to the full RALLY propagator; and finally, we analyze convergence to the Haar measure in the continuous case (a) and the discrete case (b).

\begingroup
\setcounter{secnumdepth}{4}
\renewcommand\theparagraph{\arabic{paragraph}}

\setcounter{paragraph}{0}

\paragraph{Single-layer propriety:} 

We treat the layer durations $\{\tau_\ell\}$ and pulse amplitudes $\{u^{(\ell,p)}\}$ as random variables that induce a probability measure $\nu_\ell$ on the unitary group $G=\U(d)$ for the $\ell$-th layer via the map
\begin{equation}
(\tau_\ell, \{u^{(\ell,p)}\}_p) \;\longmapsto\; U_{\ell}.
\end{equation}
Pulses within a layer are not independent because they all share the same random layer duration $\tau_\ell$. The propagator for each pulse can be written as
\begin{equation}
    U_{\ell,p}=
\exp\biggl[
  -i\,\frac{1}{N_\text{P}}
  \Bigl(
    {\tau_\ell\cal H}_{0}
    + \tau_\ell \,u^{(\ell,p)}{\cal H}_{c}
  \Bigr)
\biggr],
\end{equation}
so the relevant random parameters entering the pulses are the products $\tau_\ell u^{(\ell,p)}$. Even if the amplitudes $\{u^{(\ell,p)}\}_p$ are sampled independently, the products $\{\tau_\ell\,u^{(\ell,p)}\}_p$ are statistically dependent because they all contain the same random factor $\tau_\ell$. To arrive at a situation in which the individual pulses are statistically independent, we consider conditional probabilities with the condition on the value of the layer duration. We define the measure $\mu(\tau_\ell)$ on $G$, which is the distribution of single-pulse propagators $U_{\ell,p}$, when sampling only over the random amplitudes $u^{(\ell,p)}$ with the duration held fixed at $\tau_\ell$. Conditional on $\tau_\ell$, the pulses $\{U_{\ell,p}\}_{p=1}^{N_\text{P}}$ are i.i.d.\ with distribution $\mu(\tau_\ell)$.

We further define the conditional measure of the $\ell$-th layer, $\nu(\tau_\ell)$, as the distribution of $U_\ell(\tau_\ell)$ at fixed $\tau_\ell$. Since a layer is the product of $N_\text{P}$ pulses, we have
\begin{equation}
\label{eq:layer_measure}
  \nu(\tau_\ell)
  \;=\;
  \mu(\tau_\ell)^{*N_\text{P}},
\end{equation}
where $*$ denotes convolution on $G$ (see Ref.~\cite{Barut_irrep}).

In analogy to Ref.~\cite{Haar_convergence}, we work with the Fourier coefficients of these measures. Let $\widehat{G}$ denote the set of inequivalent irreducible unitary representations (irrep) of $G$, and, for $s\in\widehat{G}$, let
\begin{equation}
D^{(s)} : G \to \U(d_s)
\end{equation}
denote a representative of the $s$-th irrep with dimension $d_s$. The (matrix-valued) Fourier coefficient of a probability measure $f$ is defined by
\begin{equation}
\label{eq:F_coeffs}
  \hat f_s
  \;:=\;
  d_s \int_G \overline{D^{(s)}}(U)\,f(U) dU \;\in \mathbb{C}^{d_s\times d_s},
\end{equation}
where $dU$ denotes the Haar measure and $\overline{D^{(s)}}$ denotes the element-wise complex conjugation of $D^{(s)}$. The hat indicates the Fourier coefficients in the following. For convolution powers, we have the standard identity (see Eqs.~(8)–(9) of Ref.~\cite{Haar_convergence})
\begin{equation}
  \widehat{\mu^{*m}}_s
  \;=\;
  d_s\Bigl(\frac{\hat\mu_s}{d_s}\Bigr)^m.
  \label{eq:conv-power-fourier}
\end{equation}
Applying this to the conditional layer measure of Eq.~\eqref{eq:layer_measure} yields
\begin{equation}
  \widehat{\nu(\tau_\ell)}{}_s
  \;=\;
  \widehat{\mu(\tau_\ell)^{*N_\text{P}}}{}_s
  \;=\;
  d_s\Bigl(\frac{\widehat{\mu(\tau_\ell)}{}_s}{d_s}\Bigr)^{N_\text{P}},
  \label{eq:nu-tau-fourier}
\end{equation}
where
\begin{equation}
  \widehat{\mu(\tau_\ell)}{}_s
  \;:=\;
  d_s\int_G \overline{D^{(s)}}(U)\,\mu(\tau_\ell)(U) dU
\end{equation}
are the Fourier coefficients of the single-pulse distribution at duration $\tau_\ell$.

We now consider the unconditional probability measure of a layer obtained by averaging over the random layer duration $\tau_\ell$. For notational simplicity, we drop the layer index and write $\tau \equiv \tau_\ell$. Let $\rho_\tau$ denote the probability distribution of $\tau$, and define the resulting measure of a single layer as
\begin{equation}
  \nu
  \;:=\;
  \mathbb{E}_{\tau}\!\left[\nu(\tau)\right]
  \;=\;
  \int \nu(\tau)\,\rho_\tau(\tau)\,d\tau.
\end{equation}
Taking Fourier coefficients and using linearity of the integral, we obtain
\begin{equation}
  \hat\nu_s
  \;=\;
  \int \widehat{\nu(\tau)}{}_s\,\rho_\tau(\tau)d\tau
  \;=\;
  \mathbb{E}_{\tau}\!\left[
    d_s\Bigl(\frac{\widehat{\mu(\tau)}{}_s}{d_s}\Bigr)^{N_\text{P}}
  \right].
  \label{eq:nu-fourier}
\end{equation}

Now, we bound the operator norm of $\hat\nu_s$. Using the Jensen inequality and the submultiplicativity $\|AB\|\le\|A\|\,\|B\|$, we find
\begin{align}
  \frac{\|\hat\nu_s\|}{d_s}
  &\le
  \mathbb{E}_{\tau}\!\left[
    \left\|
      \Bigl(\frac{\widehat{\mu(\tau)}{}_s}{d_s}\Bigr)^{N_\text{P}}
    \right\|
  \right]
  \nonumber\\
  &\le
  \mathbb{E}_{\tau}\!\left[
    \Bigl(\tfrac{\|\widehat{\mu(\tau)}{}_s\|}{d_s}\Bigr)^{N_\text{P}}
  \right].
  \label{eq:nu-norm-bound}
\end{align}

By the main lemma of Ref.~\cite{Haar_convergence} (see the discussion leading to their Eq.~(10)), if the support of a probability measure $f$ generates a dense subgroup of $G$, then for each nontrivial irrep $s\neq 0$ we have
\begin{equation}
  \|\hat f_s\| < d_s.
  \label{eq:lemma-EL}
\end{equation}
We now assume full controllability for the RALLY propagator: there exists $0<\epsilon\le\tau_{\max}<\infty$ such that, for every $\tau\in[\epsilon,\tau_{\max}]$, the support of $\mu(\tau)$ generates a dense subgroup of $G=\U(d)$. The restriction to strictly positive durations excludes the trivial case in which all pulses with $\tau = 0$ reduce to the identity. The upper bound $\tau_{\max}$ is also physically motivated, since pulse durations must be finite in any realistic setting. Under this assumption, Eq.~\eqref{eq:lemma-EL} implies that for every $\tau\in[\epsilon,\tau_{\max}]$ and every nontrivial irrep $s\neq 0$,
\begin{equation}
  \|\widehat{\mu(\tau)}{}_s\| < d_s.
\end{equation}
Moreover, the map $\tau\mapsto\|\widehat{\mu(\tau)}{}_s\|$ is continuous on the compact interval $[\epsilon,\tau_{\max}]$, since it is the composition of continuous maps. Therefore, for each nontrivial $s$ we can define
\begin{equation}
  \alpha_s
  \;:=\;
  \max_{\tau\in[\epsilon,\tau_{\max}]}\|\widehat{\mu(\tau)}{}_s\|
  \;<\;
  d_s,
\end{equation}
which is independent of $\tau$. Inserting this into Eq.~\eqref{eq:nu-norm-bound} and defining \(\beta_s:=\alpha_s/d_s\) results for every nontrivial irrep $s$ in
\begin{equation}
  \frac{\|\hat\nu_s\|}{d_s}
  \;\le\;
  \beta_s^{N_\text{P}},
  %\;<\; 1.
  \label{eq:nu-gap}
\end{equation}
with $\beta_s<1$.

\paragraph{From single layers to the full RALLY propagator:}
Now consider the full RALLY propagator $U_{\mathrm R}$, which is a product of $N_\text{L}$ i.i.d.\ layers, each distributed according to $\nu$. The induced measure on $G$ is the $N_\text{L}$-fold convolution $\nu^{*N_\text{L}}$. Applying Eq.~\eqref{eq:conv-power-fourier} to this convolution gives
\begin{equation}
  \widehat{\nu^{*N_\text{L}}}{}_s
  \;=\;
  d_s\Bigl(\frac{\hat\nu_s}{d_s}\Bigr)^{N_\text{L}}.
\end{equation}
Therefore, for each nontrivial irrep $s\neq 0$ we obtain
\begin{align}
  \frac{\|\widehat{\nu^{*N_\text{L}}}{}_s\|}{d_s}
  &\le
  \left\|
    \frac{\hat\nu_s}{d_s}
  \right\|^{N_\text{L}}
  \nonumber\\
  &\le
  \beta_s^{N_\text{P}N_\text{L}}.
  \label{eq:full-gap}
\end{align}
Thus, all nontrivial Fourier coefficients of the distribution of $U_{\mathrm R}$ decay exponentially in the total number of pulses $N_\text{P}N_\text{L}$, with contraction factor $\beta_s$. 

\paragraph{Convergence to the Haar measure:}
In case~(a), where the control amplitudes are drawn from a continuous distribution, each $\mu(\tau)$ is a continuous probability measure on $G$ with support generating the full group. The results of Ref.~\cite{Haar_convergence} then imply that the sequence of measures $\nu^{*N_\text{L}}$ converges uniformly (in the $L^\infty$ norm) to the Haar measure on $G$, with an exponential rate governed by the bounds in Eq.~\eqref{eq:full-gap}. \hfill $\square$

In case~(b), where the control amplitudes are drawn from a discrete distribution (so that each $\mu(\tau)$ may be a finite sum of delta functions), the same Fourier analysis applies. Uniform convergence of densities is no longer available, but we still obtain exponential convergence for all polynomial test functions, as we now show explicitly. For any polynomial $\Phi:G\to\mathbb C$ define the distance
\begin{equation}
\label{eq:Delta(phi)}
  \Delta(\Phi)
  :=\int_G\Phi(U)\,\nu^{*N_{\mathrm L}}(U)dU-\int_G\Phi(U)\,dU,
\end{equation}
so that $\Delta(\Phi)\to 0$ signifies exact convergence in this weak topology. Following the Peter-Weyl theorem~\cite{Barut_irrep}, since $\Phi$ is a polynomial, there is a finite set $S\subset\widehat G$ and coefficients
$c^{(s)}_{jk}\in\mathbb C$ with
\begin{equation}\label{eq:Phi-PW}
  \Phi(U)=\sum_{s\in S}\sum_{j,k=1}^{d_s} c^{(s)}_{jk}\,D^{(s)}_{jk}(U).
\end{equation}
This expansion expresses $\Phi$ as a finite linear combination of matrix coefficients of irreducible representations, which play the role of Fourier modes on the group.

We can rewrite the first integral of Eq.~\eqref{eq:Delta(phi)} by using the finite Peter-Weyl expansion and the Fourier coefficients of \(\nu^{*N_{\mathrm L}}\),
\begin{equation}\label{eq:fourier-entry}
  \int_G D^{(s)}_{jk}(U)\,\nu^{*N_{\mathrm L}}(U)dU
  =\frac{1}{d_s}\,\overline{\bigl(\widehat{\nu^{*N_{\mathrm L}}}_s\bigr)_{jk}}.
\end{equation}
For the second integral of Eq.~\eqref{eq:Delta(phi)}, $\int_G D^{(s)}_{jk}(U)\,dU=0$ for every nontrivial $s$ because of orthogonality. Moreover, the term with $s=0$ is equal for both integrals, and then 
\begin{equation}\label{eq:Delta-expansion}
  \Delta(\Phi)
  =\sum_{\substack{s\in S\\ s\neq 0}}\sum_{j,k=1}^{d_s}
  c^{(s)}_{jk}\,\frac{1}{d_s}\,\overline{\bigl(\widehat{\nu^{*N_{\mathrm L}}}_s\bigr)_{jk}}.
\end{equation}
% Absolute-value bound
Taking absolute values and using $|(M)_{jk}|\le\|M\|$ gives
\begin{equation}\label{eq:Delta-bound1}
  |\Delta(\Phi)|
  \le\sum_{\substack{s\in S\\ s\neq 0}}\sum_{j,k=1}^{d_s}
  |c^{(s)}_{jk}|\frac{\bigl\|\widehat{\nu^{*N_{\mathrm L}}}_s\bigr\|}{d_s}.
\end{equation}
% Spectral-gap -> exponential bound
Using inequality of Eq.~\eqref{eq:full-gap}, we obtain
\begin{equation}\label{eq:Delta-bound2}
  |\Delta(\Phi)|
  \le\sum_{\substack{s\in S\\ s\neq 0}}\sum_{j,k=1}^{d_s}
  |c^{(s)}_{jk}|\,\beta_s^{\,N_{\mathrm P}N_{\mathrm L}}.
\end{equation}
Since the index set $S$ is finite, we may introduce the constants
\begin{equation}
  C_\Phi
  := \sum_{\substack{s\in S\\ s\neq 0}}\sum_{j,k=1}^{d_s}\bigl|c^{(s)}_{jk}\bigr|
  \qquad\text{and}\qquad
  \beta
  := \max_{\substack{s\in S\\ s\neq 0}}\beta_s < 1.
\end{equation}
With this notation, we obtain the compact bound
\begin{equation}
\label{eq:final-bound}
  |\Delta(\Phi)|
  \le C_\Phi\,\beta^{\,N_{\mathrm P}N_{\mathrm L}}.
\end{equation}
Thus, the discrepancy $|\Delta(\Phi)|$ decays exponentially in $N_{\mathrm P}N_{\mathrm L}$ with respect to the weak topology induced by polynomial test functions, cf. Eq.~\eqref{eq:Delta(phi)}. In the language of unitary designs, this implies exponential convergence of the RALLY ensemble to an approximate unitary $t$-design. \hfill $\square$

In summary, under the assumption of controllability, the distribution of the RALLY propagator $U_{\mathrm R}$ converges exponentially fast to the Haar measure on $\U(d)$: uniformly in the continuous-amplitude case, and in the sense of all degree-$t$ polynomial test functions (i.e., as an approximate unitary $t$-design) in the discrete-amplitude case.

\endgroup

\section{Numerical evidence of convergence to Haar ensemble} \label{App:Harr_conv}
In this section, we present a numerical criterion to provide evidence of convergence to the Haar measure on $\operatorname{U}(d)$. The criterion involves convergence to $t$-unitary designs via a scalar inequality (Theorem~5.4 of Ref.~\cite{Scott_2008}), originally proved for finite weighted sets and here extended to general probability measure. It is convenient for numerical calculations because it depends only on pairwise overlaps of unitaries and can be estimated directly from samples.

Let us define the positive constant
\begin{equation}\label{gamma-def}
\gamma(t,d):=\int_{\operatorname{U}(d)}\!dU\;\big|\! \operatorname{tr}[U] \big|^{2t}\,.
\end{equation}
If $t\leq d$, this integral is equal to $t!$ (see Ref.~\cite{Scott_2008}).
\begin{theorem}[Finite uniform-weight $t$-design bound] \label{th:fu}
Let $X\subset\operatorname{U}(d)$ be a finite set with sample size $N_\text{x}:=|X|$, and consider $X$ with the uniform weights
\begin{equation}
w(x)=\frac{1}{N_\text{x}}\qquad(x\in X).
\end{equation}
Then for any integer $t\ge 1$,
\begin{equation}\label{finite-bound}
\frac{1}{N^{2}_\text{x}}\sum_{x,y\in X}\Big| \operatorname{tr}\big[U(x)^\dagger U(y)\big]\Big|^{2t}
\;\ge\; \gamma(t,d),
\end{equation}
with equality if and only if $X$ is a $t$-design.
\end{theorem}
\begin{proof}[Proof: (See Ref.~\cite{Scott_2008})]
\end{proof}

\begin{theorem}[Continuous-measure variant of Theorem~\ref{th:fu}]
Let $\nu$ be any probability measure on $\operatorname{U}(d)$, and let $t\ge 1$.
Define
\begin{equation}\label{def-Ft}
F_t(\nu)
:=\iint_{\operatorname{U}(d)\times\operatorname{U}(d)}
\Big| \operatorname{tr}\big[U(x)^\dagger U(y)\big]\Big|^{2t}\,d\nu(x)\,d\nu(y).
\end{equation}
Then
\begin{equation}\label{continuous-bound}
F_t(\nu)\;\ge\;\gamma(t,d),
\end{equation}
with equality if and only if $\nu$ is a unitary $t$-design.
\end{theorem}

\begin{proof}[Proof: (Same as in Ref.~\cite{Scott_2008}: replace sums by integrals)]
\end{proof}

The double integral $F_t(\nu)$ can be approximated numerically by Monte-Carlo sampling: draw \(M\) independent pairs \((U_i,V_i)\) from the product measure \(\nu\otimes\nu\) and form the unbiased estimator
\begin{equation}\label{mc-estimator}
F_{t,M}:=\frac{1}{M}\sum_{i=1}^{M} X_i
\quad\text{with}\quad
X_i:=\big|\!\operatorname{tr}(U_i^\dagger V_i)\big|^{2t}.
\end{equation}
The strong law of large numbers implies \(F_{t,M}\xrightarrow{}F_t(\nu)\) as \(M\to\infty\), and the resulting distribution can be approximated according to the central limit theorem by a Gaussian distribution $\mathcal{N}$
\begin{equation}
F_{t,M}\approx\mathcal{N}\Big(F_t(\nu),\frac{\sigma_t^2}{M}\Big),
\end{equation}
with mean $F_t(\nu)$ and variance $\sigma_t^2/M=\operatorname{Var}\big(\,|\! \operatorname{tr}(U^\dagger V)|^{2t}\,\big)/M$ for large \(M\). Finite sampling may produce a violation of the inequality in Eq.~\eqref{continuous-bound} due to sampling noise.  A robust numerical quantity is, therefore, the absolute deviation
\begin{equation}
\delta_{t}:=\big| F_{t,M}-\gamma(t,d)\big|.
\end{equation}
The sampling fluctuations of \(\delta_{t}\) are of the order \(\sigma_t / \sqrt{M}\). Thus, when plotting \(\delta_{t}\), one expects a plateau of the order of this statistical uncertainty. If \(U\) and \(V\) are independently sampled from the Haar measure on \(\mathrm{U}(d)\), with \(d \geq 4\) as in the main text, one finds (see Ref.~\cite{collins2002momentscumulantspolynomialrandom})
\begin{align}
\operatorname{Var}\!\left(\left|\operatorname{tr}\big(U^{\dagger}V\big)\right|^{2}\right) &= 1, \\
\operatorname{Var}\!\left(\left|\operatorname{tr}\big(U^{\dagger}V\big)\right|^{4}\right) &= 20, \\
\operatorname{Var}\!\left(\left|\operatorname{tr}\big(U^{\dagger}V\big)\right|^{6}\right) &= 658, \\
\operatorname{Var}\!\left(\left|\operatorname{tr}\big(U^{\dagger}V\big)\right|^{8}\right) &= 32748.
\end{align}

\section{Convergence to Haar ensemble for individual random choices of pulse amplitudes} \label{app_typ}
As shown in Fig.~\ref{fig:conv_plot}, $\delta_t$ exhibits the same exponential decay for each individual random choice of the pulse amplitudes when compared to the decay of $\delta_t$ averaged over random choices of the pulse amplitudes. The small fluctuations relative to the exponential decay for the individual random choices of the pulse amplitudes were attributed to the finite size of the spin system. Indeed, when considering only two instead of three spins in the Ising spin chain, the fluctuations significantly increase as shown in Fig.~\ref{fig:conv_plot_2}. Hence, we anticipate that, for an Ising spin chain of $n=6$ spins considered in the main text, these fluctuations become negligible.

\begin{figure}[h]
  \centering
  \begin{overpic}[width=0.9\columnwidth]{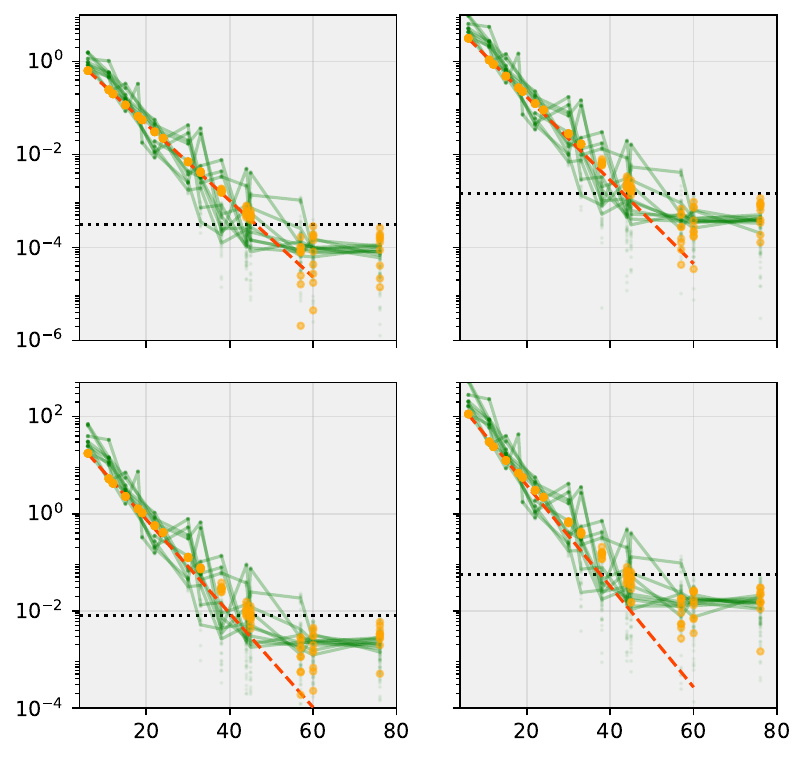}
    % Common x-axis label (centered at bottom)
    \put(30,-1){\makebox(0,0){\small $N_\text{L}N_\text{P}$}}
    \put(78,-1){\makebox(0,0){\small $N_\text{L}N_\text{P}$}}

    \put(-2.5,70){{\small $\delta_1$}}
    \put(52,70){{\small $\delta_2$}}
    \put(-2.5,25){{\small $\delta_3$}}
    \put(52,25){{\small $\delta_4$}}

    \put(60,40){%
      \begin{tikzpicture}
        \node[
          fill=white,
          fill opacity=0.8,
          text opacity=1,
          draw=black,
          rounded corners=2pt,
          inner sep=2pt
        ] {
          \begin{tikzpicture}[x=1em,y=1em]
    
            \draw[orange, thick,
                  mark=*,
                  mark size=1.5pt, 
                  mark options={solid}]
              plot coordinates {(0.9,3)};
            \node[right=0.6em] at (1.8,3) {\scriptsize $\delta_t$};

            \draw[myGreen] (0,2) -- (1.8,2);
            \draw[myGreen,
                  mark=*,
                  mark size=1.5pt, 
                  mark options={solid}]
              plot coordinates {(0.9,2)};
            \node[right=0.6em] at (1.8,2) {\scriptsize $\delta_t$ with fixed $u^{(\ell, p)}$};

            \draw[orangered, thick, dashed] (0,1) -- (1.8,1);
            \node[right=0.6em] at (1.8,1) {\scriptsize Exponential fit};

            \draw[black, thick, dotted] (0,0) -- (1.8,0);
            \node[right=0.6em] at (1.8,0) {\scriptsize Expected plateau};
          \end{tikzpicture}
        };
      \end{tikzpicture}%
    }

  \end{overpic}
  \caption{Same as Fig.~\ref{fig:conv_plot} but for an Ising spin chain with two spins: Exponential convergence of the first four moments of the ensemble of unitaries reachable by the RALLY$_\text{T}$ propagator $U_\text{R,T}$ to the corresponding moments of a Haar-random distribution. The difference $\delta_t$~\eqref{eq:delta_t} between the corresponding moments is shown for $t=1,2,3,4$ as a function of $N_{\mathrm{L}}N_{\mathrm{P}}$. Orange points show $\delta_t$ with layer durations drawn uniformly from $[0,10]$ and control amplitudes drawn uniformly from $[-1,1]$. The dashed red line is a linear fit on the logarithmic scale. More importantly, results for $\delta_t$ follow the same decay curve also when \textit{only layer durations} are sampled and the control amplitudes are fixed (shown by the 10 green lines in each panel). We attribute the fluctuations of the green lines relative to the exponential decay to the finite system sizes. The plateaus at larger values of $N_{\mathrm{L}}N_{\mathrm{P}}$ are due to the finite sample size and they are consistent with the expected values from the Haar distribution (black dotted lines); see Appendix~\ref{App:Harr_conv} for details. For each value of $N_{\mathrm{L}}N_{\mathrm{P}}$, where $N_{\mathrm{L}}\in\{6,11,15,19\}$ and $N_{\mathrm{P}}\in\{1,2,3,4\}$, we use $10^8$ samples and average over 10 repetitions for one particular random choice of values of $h_{ix}$ and $h_{iz}$, cf. Eq.~\eqref{eq:random_spin_H}. }
  \label{fig:conv_plot_2}
\end{figure}
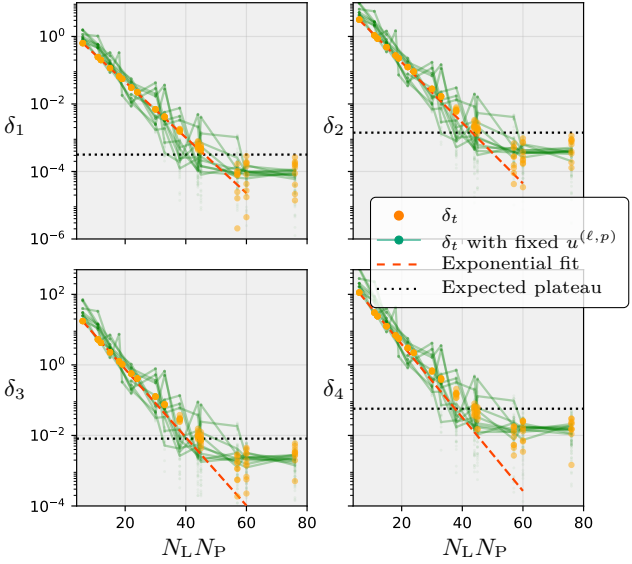

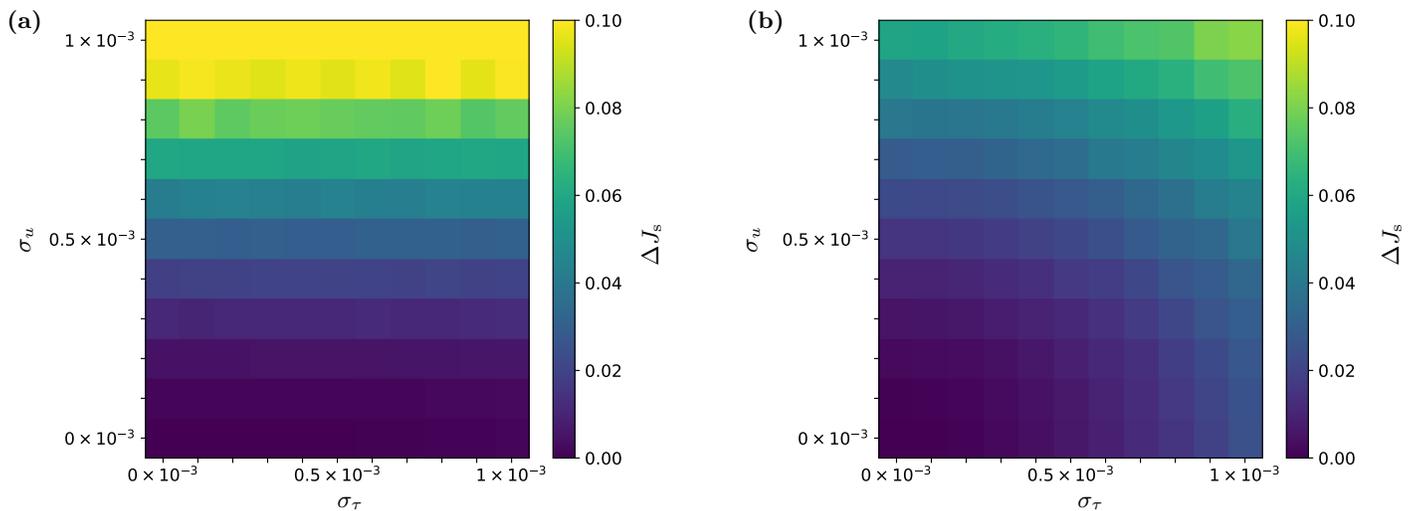
\begin{figure*}[bt]
  \centering
  \includegraphics[width=0.45\textwidth]{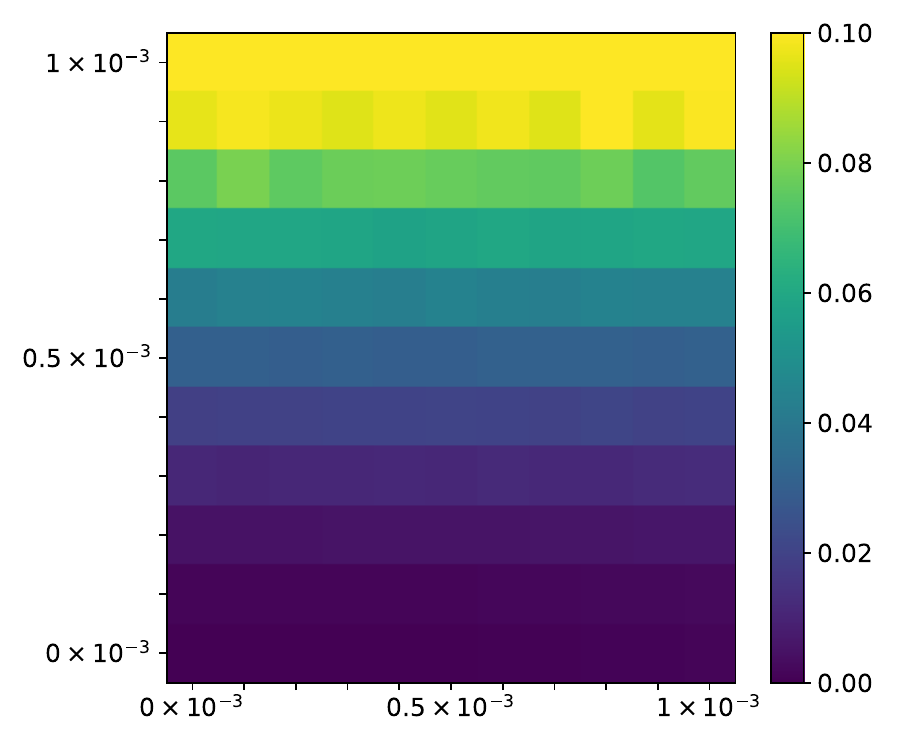}
  \hfill
  \includegraphics[width=0.45\textwidth]{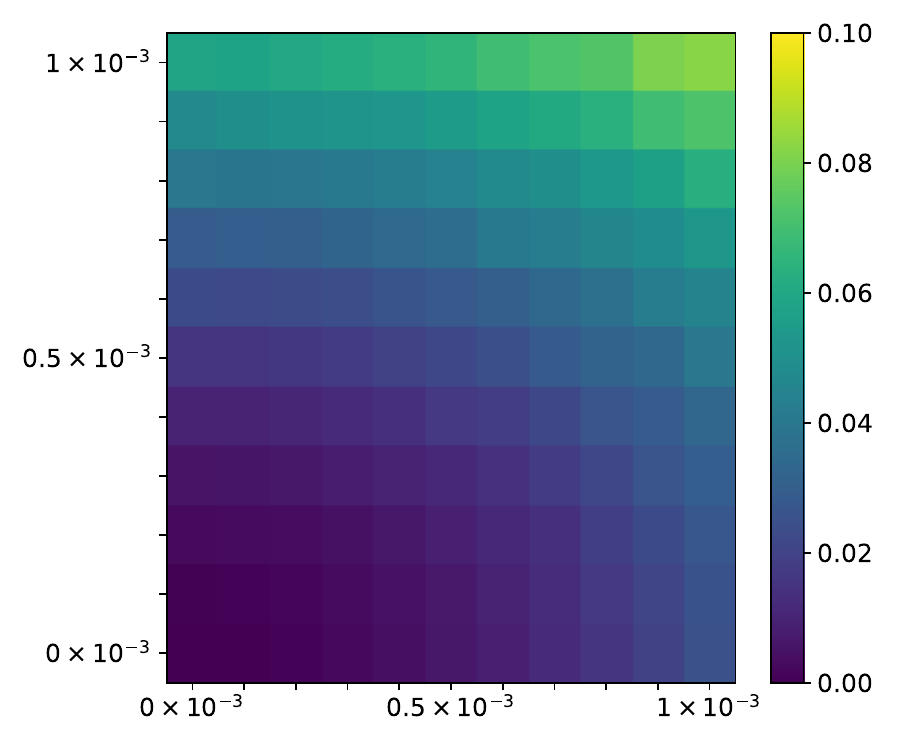}
  \begin{picture}(0,0)
    % panel labels
    \put(-520,180){\textbf{(a)}}
    \put(-240,180){\textbf{(b)}}

    % x-axes: sigma_tau
    \put(-390, 0){\makebox(0,0){$\sigma_\tau$}} % left plot
    \put(-110, 0){\makebox(0,0){$\sigma_\tau$}} % right plot

    % left y-axes: sigma_u
    \put(-515, 95){\rotatebox{90}{$\sigma_u$}}   % left plot
    \put(-240, 95){\rotatebox{90}{$\sigma_u$}}   % right plot

    % right y-axes: ΔJ_s
    \put( 0, 90){\rotatebox{90}{$\Delta J_\text{s}$}}   % left plot
    \put( -280, 90){\rotatebox{90}{$\Delta J_\text{s}$}}   % right plot
    
  \end{picture}

  \caption{Mean change of the state-transfer infidelity $\Delta J_\text{s}$ for GRAPE~(a) and RALLY$_\text{T}$~(b) on a 6-spin Ising spin chain under Gauss-distributed perturbations of control amplitudes $\delta u^{(\ell,p)}$ and pulse durations $\delta \tau_\ell$. Each point is averaged over 10 perturbation for each of the 10 different random values of $h_{ix}$ and $h_{iz}$. When compared to solutions of the GRAPE algorithm, the solutions of RALLY$_\text{T}$ appear more robust to deviations of the pulse amplitudes but less robust to deviations of the time durations due to the layer size $N_\text{P}>1$, as described in the text.}
  \label{fig:robustness}
\end{figure*}

\section{Robustness of solutions of RALLY$_\text{T}$} \label{app:roboustness}
In this section, we analyze how perturbations of pulse durations and control amplitudes affect the optimized solutions of RALLY$_\text{T}$. We first derive a perturbative bound on the deviation of the total propagator and then relate this deviation to the change in state-transfer infidelity~\eqref{eqn_state-trans}. Finally, we compare these predictions with the results of numerical simulations of the Ising spin chain~\eqref{eq:random_spin_H}.

Let us consider the RALLY$_\text{T}$ propagator introduced in Eq.~\eqref{eq:RALLY_T}
\begin{equation} \label{eqn_copy_RT}
    U_\text{R,T} = \prod_{\ell=1}^{N_\text{L}}\prod_{p=1}^{N_\text{P}} U_{\ell,p}
\end{equation}
with $U_{\ell,p} = \exp\left[-i\frac{\tau_\ell}{N_\text{P}}{\cal H}^{(\ell,p)}\right]$ and
\begin{equation} \label{eqn_hlp}
    {\cal H}^{(\ell,p)}\equiv{\cal H}_{0} + u^{(\ell,p)}{\cal H}_{c}.
\end{equation}
For the perturbation of the parameters $\tau_\ell' = \tau_\ell + \delta \tau_\ell$ and $ {u^{(\ell,p)}}' = u^{(\ell,p)} + \delta u^{(\ell,p)}$, we define the perturbed propagators
\begin{equation}
    U_{\ell,p}' = \exp\biggl[-i\frac{\tau_\ell'}{N_\text{P}}\bigl({\cal H}_{0} + {u^{(\ell,p)}}'{\cal H}_{c}\bigr)\biggr],
\end{equation}
such that, for small perturbations, we can write
\begin{equation}
    \Delta U_{\ell,p} \equiv U_{\ell,p}' - U_{\ell,p}
    \approx \frac{\partial U_{\ell,p}}{\partial \tau_\ell}\delta \tau_\ell + \frac{\partial U_{\ell,p}}{\partial u^{(\ell,p)}}\delta u^{(\ell,p)}.
\end{equation}
For the total propagator $U_\text{R,T}$ in Eq.~\eqref{eqn_copy_RT}, we obtain for the average Hilbert-Schmidt norm of the deviation
\begin{eqnarray}\label{eq_eulp}
&&\mathbb{E}\|\Delta U_\text{R,T}\|_{HS} \\
    &&\le\sqrt{\frac{2}{\pi}}
        \left[
        \frac{\sigma_\tau}{N_\text{P}}\sum_{\ell=1}^{N_\text{L}}\sum_{p=1}^{N_\text{P}} \|{\cal H}^{(\ell,p)}\|_{HS} +\sigma_u  \|{\cal H}_c\|_{HS}\sum_{\ell=1}^{N_\text{L}} \tau_\ell \right],\nonumber
\end{eqnarray}
where we assumed Gauss-distributed perturbations of $\tau_\ell$ with standard deviation $\sigma_\tau$ and of $u^{(\ell,p)}$ with standard deviation $\sigma_u$. The same expression with $N_\text{P}=1$ applies to the GRAPE propagator.

The change in state-transfer infidelity~\eqref{eqn_state-trans} due to the perturbations is
\begin{eqnarray}
    \Delta J_\text{s} \equiv J_\text{s}(U_\text{R,T}') - J_\text{s}(U_\text{R,T})
= J_\text{s}(U_\text{R,T}'),
\end{eqnarray}
where we assumed that $J_\text{s}(U_{\ell,p})=0$, i.e., the optimization found the optimal solution very accurately, and
\begin{equation}
    J_\text{s}(U_\text{R,T}')= 1 - \bigl|\braket{\psi_{\mathrm{T}}|{U_\text{R,T}'|\psi_0}}\bigr|^2,
\end{equation}
with $\ket{\psi_0}$ being the initial state and $\ket{\psi_{\mathrm{T}}}$ the target state.

Writing $U_\text{R,T}' = U_\text{R,T} + \Delta U_\text{R,T}$ and expanding to first order in $\Delta U_\text{R,T}$, one finds
\begin{equation}
|\Delta J_\text{s}|\lesssim 2\,\bigl|\bra{\psi_0}U_\text{R,T}^\dagger \Delta U_\text{R,T}\ket{\psi_0}\bigr|
\le 2\,\|\Delta U_\text{R,T}\|_{HS}.
\end{equation}
This implies that $\mathbb{E}\,\|\Delta U_\text{R,T}\|_{HS}$ in Eq.~\eqref{eq_eulp} bounds the change of the state-transfer infidelity from above.

We validate the perturbative bound with numerical simulations of the state-transfer problem on the 6-spin Ising spin chain (see Sec.~\ref{sec_res}) using both GRAPE and RALLY$_\text{T}$. For each pair of standard deviations $(\sigma_u,\sigma_\tau)$, we generate a set of 10 different values of $h_{ix}$ and $h_{iz}$, optimize the pulse sequences with GRAPE and RALLY$_\text{T}$ to minimize $J_\text{s}$ and then add independent Gaussian perturbations to every pulse ($\delta u^{(\ell,p)}\sim\mathcal{N}(0,\sigma_u)$ and $\delta \tau_\ell\sim\mathcal{N}(0,\sigma_\tau)$ with $\mathcal{N}(0,\sigma)$ being a Gaussian distribution with mean 0 and standard deviation $\sigma$). We obtain $\Delta J_\text{s}$ and average it over 10 different perturbations for each of the 10 values of $h_{ix}$ and $h_{iz}$. The results are shown in Fig.~\ref{fig:robustness}.

For RALLY$_\text{T}$, the contribution of the control-amplitude uncertainty to $\mathbb{E}\|\Delta U_\text{R,T}\|_{HS}$ (last term in Eq.~\eqref{eq_eulp}) depends only on the total time duration of the pulse-sequence  $\sum_{\ell=1}^{N_\text{L}} \tau_\ell$, but not on the total number of pulses $N_\text{L}N_\text{P}$. Thus, GRAPE and RALLY$_\text{T}$ with similar total pulse-sequence time are expected to behave similarly with respect to amplitude uncertainty independent of the layer size \(N_\text{P}\). This is corroborated by the results shown in Fig.~\ref{fig:robustness}, where solutions of RALLY$_\text{T}$  indicate a slightly higher robustness than solutions of the GRAPE algorithm to perturbations of the pulse amplitudes. In contrast, the contribution of the time uncertainty to $\mathbb{E}\|\Delta U_\text{R,T}\|_{HS}$ (first term in Eq.~\eqref{eq_eulp}) is proportional to the total number of pulses $N_\text{L}N_\text{P}$. Thus, the robustness of the solutions of RALLY$_\text{T}$  is expected to be reduced by the factor $N_\text{P}>1$, which is confirmed by the results shown in Fig.~\ref{fig:robustness}.

\section{Experimental constraints for the globally-driven Rydberg platform} \label{app:constraints}
For the numerical simulations of the globally-driven Ryd\-berg platform in Sec.~\ref{sec_res}, we use the physical constraints of the PASQAL machine operating with Rubidium atoms~\cite{Silverio2022pulseropensource}. These constraints are detailed in Table~\ref{table:exp_cons}.

\begin{table}[h]
    \centering
    \begin{tabular}{@{}ll@{}}
        \toprule
        Parameter & Value \\ 
        \midrule
        Rydberg level $n$ & 70 \\ 
        Minimum Rabi frequency $\Omega_\text{min}$ & 0 MHz \\
        Maximum Rabi frequency $\Omega_\text{max}$ & 10 MHz \\
        Minimum Detuning $\Delta_\text{min}$ & -10 MHz \\
        Maximum Detuning $\Delta_\text{max}$ & 10 MHz \\
        Minimum distance between atoms & 5 $\mu$m \\
        Minimum pulse duration & 4 ns \\
        Maximum pulse duration & $10^5$ $\mu$s \\
        Interaction coefficient $C_{6,rr}$ & 5420 GHz $\mu m^6$ \\ 
        \bottomrule
    \end{tabular}
    \caption{Physical constraints of the PASQAL machine operating with Rubidium atoms~\cite{Silverio2022pulseropensource} considered in this article.}
    \label{table:exp_cons}
\end{table}

\section{Details of unitary-synthesis optimization} \label{app:CNOT_coordinates}
We choose a non-symmetric spatial Rydberg-atom configuration for the unitary-synthesis problem. The positions of the atoms are (in \(\mu\)m): $\mathbf{r}_{q_0} = (-8,\; 0)$, \mbox{$\mathbf{r}_{q_1} = (0,\; 9.6)$}, $\mathbf{r}_{q_2} = (8.8,\; 0)$. The settings for the optimization were chosen as follows.
\paragraph{dCRAB settings:}
We use the QuOCS open-source optimal control suite~\cite{QuocsRossignolo2023} for the dCRAB simulations with the adaptive Nelder-Mead algorithm~\cite{adaptie_Nelder}. The number of super-iterations is 3 and the detuning is restricted to $\Delta \in[-10,10]\,\text{MHz}$ by setting bounds in the optimization. The final evolution time $T$ is $500\,\mu \text{s}$ and the bandwidth limitation is \(\Delta\Omega = 100/ T \) to ensure successful optimization. All optimizer settings are summarized in Table~\ref{tab:opt_settings}.

\paragraph{RALLY settings:}
The layer size is \mbox{$N_P=5$} and the detuning amplitude is restricted to \mbox{$\Delta \in[-10,10]\,\text{MHz}$}. To ensure that these constraints are satisfied during RALLY$_\text{A}$ optimizations, we enforce $\bm{\xi} \in [0,1]^{N_\text{L}}$. For RALLY$_\text{T}$, the minimum duration is set to $4\,\text{ns}$ for every pulse in the optimization and the initial layer durations are chosen from a uniform distribution in [0,1]\,$\mu$s. Optimization is performed with the Nelder–Mead method using the same settings as for the dCRAB algorithm.

\begin{table}[h!]
\centering
\small
\begin{tabular}{@{}lccc@{}}
\toprule
Parameter & dCRAB & RALLY$_\text{T}$ & RALLY$_\text{A}$\\ \midrule
Structure & 3 super-iter. & \(N_P=5\) & \(N_P=5\) \\
Detuning bounds & \multicolumn{3}{c}{\([-10,10]\,\mathrm{MHz}\)} \\
Final evolution time \(T\) & \(500\,\mu\mathrm{s}\) & — & \(500\,\mu\mathrm{s}\) \\
Bandwidth limit \(\Delta\Omega\) & \( 100/T\) & — & — \\
\texttt{xatol} & \multicolumn{3}{c}{\(10^{-8}\)} \\
\texttt{fatol} & \multicolumn{3}{c}{\(10^{-8}\)} \\
Optimizer & \multicolumn{3}{c}{adaptive Nelder--Mead} \\
Max. FoM evaluations & \multicolumn{3}{c}{\(10^{6}\)} \\
\bottomrule
\end{tabular}
\caption{Optimization settings for dCRAB and RALLY optimizations. Centered entries (spanning all columns) denote parameters shared by all methods.}
\label{tab:opt_settings}
\end{table}

\section{Details of the ground-state preparation} \label{app:VQS_optimization}
The choice of spatial configuration of the Rydberg atoms is crucial because a configuration that reflects the symmetry of the target Hamiltonian can constrain the optimization problem to a symmetry sector, thereby reducing the dimensionality of the relevant Hilbert subspace and easing the optimization. At the same time, symmetries of the spatial configuration beyond those of the Hamiltonian can prohibit the reachability of the target state, when the initial and the target states are localized in different symmetry sectors~\cite{gs-reachability}. 

\subsection{Details of the optimization for the $\text{H}_\text{2}$ molecule} \label{app:H2}
To describe the spin-orbital degrees of freedom for the H$_2$ molecule with an interatomic distance of 2 Ångström, we use the following target Hamiltonian
\begin{equation}
\begin{aligned}
\cal{H}_\text{T} =\ &
-0.81054798\,\mathds{1}
+ 0.17218393\,\sigma_{0}^{z} \\
&- 0.22575349\,\sigma_{1}^{z}
+ 0.17218393\,\sigma_{2}^{z} \\
&- 0.22575349\,\sigma_{3}^{z}
+ 0.12091263\,\sigma_{1}^{z}\sigma_{0}^{z} \\
&+ 0.16892754\,\sigma_{2}^{z}\sigma_{0}^{z}
+ 0.16614543\,\sigma_{2}^{z}\sigma_{1}^{z} \\
&+ 0.16614543\,\sigma_{3}^{z}\sigma_{0}^{z}
+ 0.17464343\,\sigma_{3}^{z}\sigma_{1}^{z} \\
&+ 0.12091263\,\sigma_{3}^{z}\sigma_{2}^{z}
+ 0.04523280\,\sigma_{3}^{x}\sigma_{2}^{x}\sigma_{1}^{x}\sigma_{0}^{x} \\
&+ 0.04523280\,\sigma_{3}^{x}\sigma_{2}^{x}\sigma_{1}^{y}\sigma_{0}^{y}
+ 0.04523280\,\sigma_{3}^{y}\sigma_{2}^{y}\sigma_{1}^{x}\sigma_{0}^{x} \\
&+ 0.04523280\,\sigma_{3}^{y}\sigma_{2}^{y}\sigma_{1}^{y}\sigma_{0}^{y},
\end{aligned}
\end{equation}
which is obtained by using the Jordan-Wigner transformation. Because of the use of the Jordan-Wigner transformation, the above Hamiltonian exhibits the symmetry of the simultaneous exchange of spins $0\leftrightarrow2$ and $1\leftrightarrow3$. Hence, we arrange the Rydberg atoms on a rhombus, cf. Fig.~\ref{Fig:H2_config}, that respects this symmetry, thereby allowing the initial state $\ket{\psi_0}$ and the corresponding ground state to belong to a fully controllable symmetry sector of $N=9$ dimensions described by $2N-2=16$ real parameters. In contrast, a symmetric ring configuration would introduce additional symmetries, which are incompatible with the Hamiltonian, thereby preventing the reachability of the H$_2$ ground state~\cite{gs-reachability}. 

\begin{figure}[t]
    \centering
    \hfill
    % First image with label (a) at the top left
    \begin{minipage}[t]{0.35\columnwidth}
        \includegraphics[width=\linewidth]{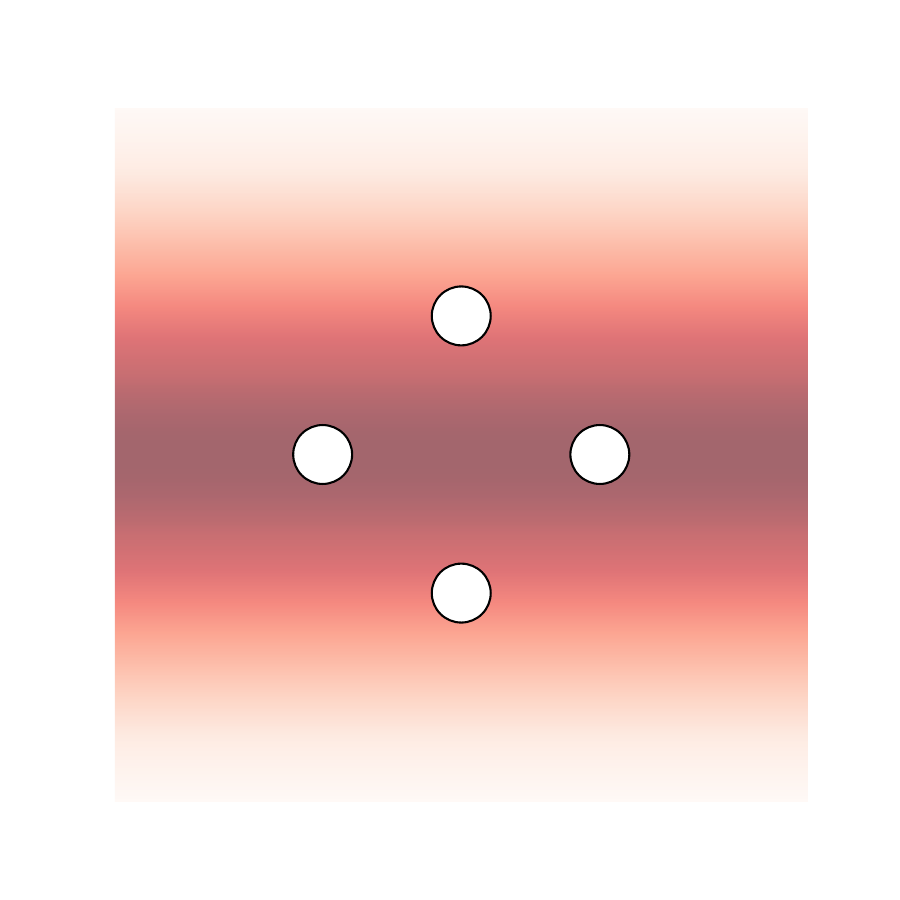}
        \begin{picture}(0,0)
            \put(-50, 90){\makebox(0,0)[lt]{\textbf{(a)}}}
        \end{picture}
    \end{minipage}
    \hfill
    % Second image with label (b) at the top left
    \begin{minipage}[t]{0.35\columnwidth}
        \includegraphics[width=\linewidth]{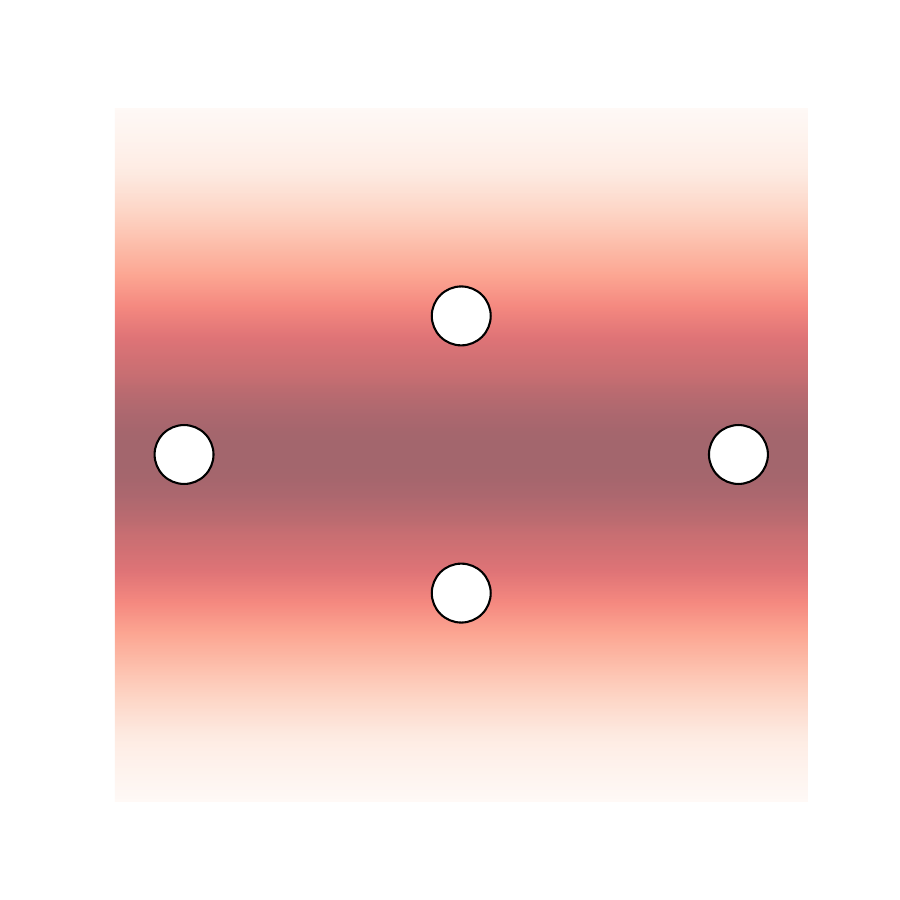}
        \begin{picture}(0,0)
            \put(-50, 90){\makebox(0,0)[lt]{\textbf{(b)}}}
        \end{picture}
    \end{minipage}
    \hfill
    % Caption for the figure
    \caption{Possible spatial configurations of the Rydberg atoms for the ground-state preparation of H$_2$ molecule. (a) The ring configuration introduces additional symmetries not compatible with the target Hamiltonian. (b) The rhombic shape reflects the exchange symmetry of the target Hamiltonian. The atomic coordinates for this configuration, ordered by spin index (in \(\mu\)m), are: (14, 0), (0, 7), (-14, 0), (0, -7).}
    \label{Fig:H2_config}
\end{figure}
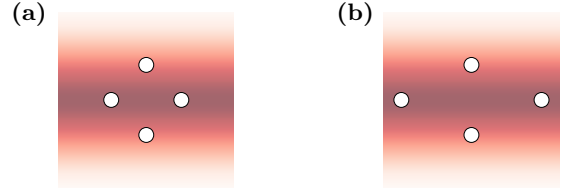

For the optimization, we used the same setting as for the unitary-synthesis problem except that we limited the maximum number of figure-of-merit evaluations to \(\leq10^5\) for all methods. For RALLY$_\text{A}$, we set the final evolution time to \(T = 100\,\mu\mathrm{s}\). For dCRAB, a further bandwidth limit of \(\Delta\Omega = 30/T\) was chosen.

\subsection{Details of the optimization for the $\text{NO}_\text{3}$ molecule} \label{app:NO3}
To describe the $\text{NO}_\text{3}$ molecule with the spatial coordinates of the individual atoms given in Table~\ref{table:NO3}, we use a target spin Hamiltonian provided in Ref.~\cite{code}. Since parity encoding~\cite{parity} was used, the four-qubit Hamiltonian does not exhibit any symmetries related to the exchange of spins. Hence, a non-symmetric Rydberg-atom configuration shown in Fig.~\ref{NO3_config} was chosen to obtain full controllability in the full Hilbert space of dimension $N=16$ described by $2N-2=30$ real parameters.
\begin{table}[h]
    \centering
    \begin{tabular}{lccc}
        \hline
        Atom & X & Y & Z \\
        \hline
        N & -0.0225 & -0.0483 & 0.0029 \\
        O & -0.6411 & 1.0149 & -0.0004 \\
        O & -0.6284 & -1.1191 & 0.0095 \\
        O & 1.2096 & -0.0414 & -0.0004 \\
        \hline
    \end{tabular}
    \caption{Spatial coordinates of atoms of the NO$_3$ molecule in Ångström used in our calculations.}
    \label{table:NO3}
\end{table}

For the optimization, we retained the parameters used for the unitary-synthesis problem. The only change was that, for dCRAB, the total evolution time was set to \mbox{\(T = 500\,\mu\mathrm{s}\)}
or \(T = 3000\,\mu\mathrm{s}\) and the bandwidth was limited to \(\Delta\Omega = 80/T\). For RALLY$_\text{A}$, we used \(T = 500\,\mu\mathrm{s}\).
\begin{figure}[h]
  \centering
  \includegraphics[width=0.2\textwidth]{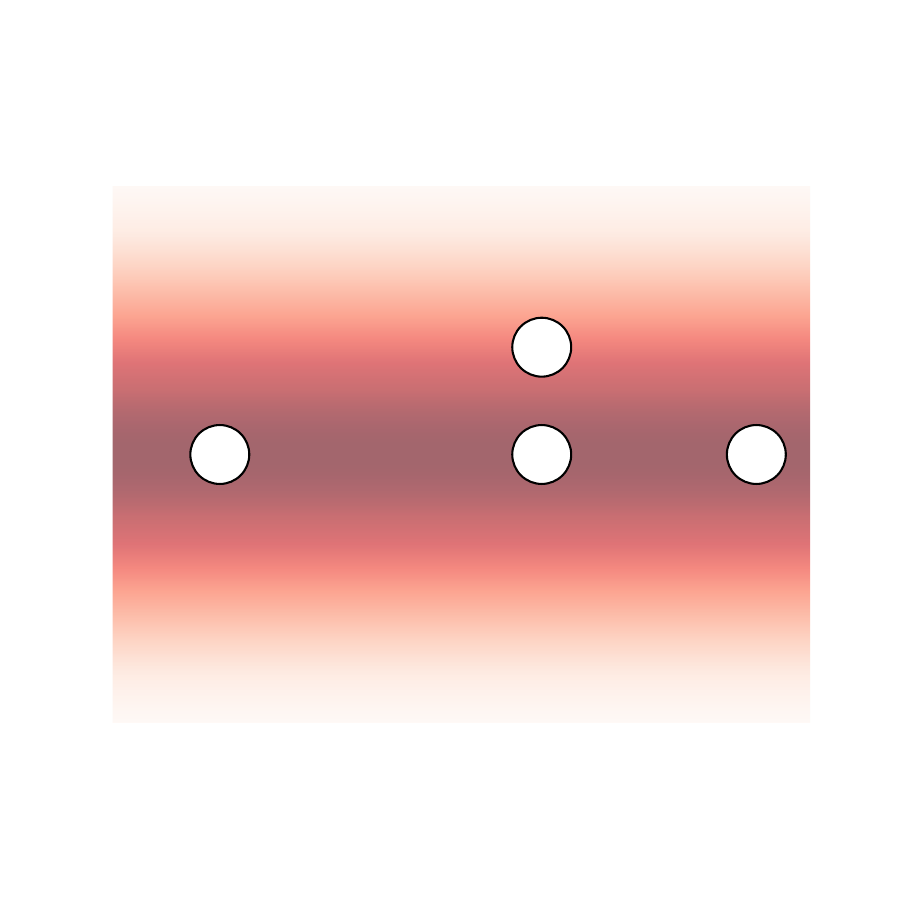}
  \caption{Spatial configuration of Rydberg atoms used for the ground-state preparation of the NO$_3$ molecule. The atomic coordinates for this configuration ordered by spin index (in $\mu$m) are: (10, 0), (0, 20), (0, 0), (0, -30).}
  \label{NO3_config}
\end{figure}

\subsection{Details of the optimization for the CH molecule}\label{app:CH}
To describe the CH molecule with an interatomic distance of 1.12388 Ångström, we use the target Hamiltonian provided in Ref.~\cite{code}. 
% \begin{table}[h]
%     \centering
%     \begin{tabular}{lccc}
%         \hline
%         Atom & X Coordinate & Y Coordinate & Z Coordinate \\
%         \hline
%         C & 0.0000 & 0.0000 & -0.00198 \\
%         H & 0.0000 & 0.0000 & 1.1219 \\
%         \hline
%     \end{tabular}
%     \caption{Coordinates of atoms of the CH molecule used in our calculation (in Ångström).}
%     \label{table:CH}
% \end{table}
Since the Jordan-Wigner transformation was used to obtain the six-qubit target Hamiltonian, the Hamiltonian exhibits a permutation symmetry similar to the symmetry of the target Hamiltonian for the H$_2$ molecule. Hence, a spatial configuration of Rydberg atoms similar to the one shown in Fig.~\ref{Fig:H2_config}(b), for example, the spatial configuration in Fig.~\ref{Fig:CH_config}(a), could seem suitable at first sight. However, for the spatial configuration of the Rydberg atoms in Fig.~\ref{Fig:CH_config}(a), we observe that the initial state $\ket{\Vec{0}}$ is in a different symmetry sector compared to the ground state of CH. Therefore, we use the non-symmetric Rydberg-atom configuration shown in Fig.~\ref{Fig:CH_config}(b).

\begin{figure}[h]
    \centering
    \hfill
    % First image with label (a) at the top left
    \begin{minipage}[t]{0.4\columnwidth}
        \includegraphics[width=\linewidth]{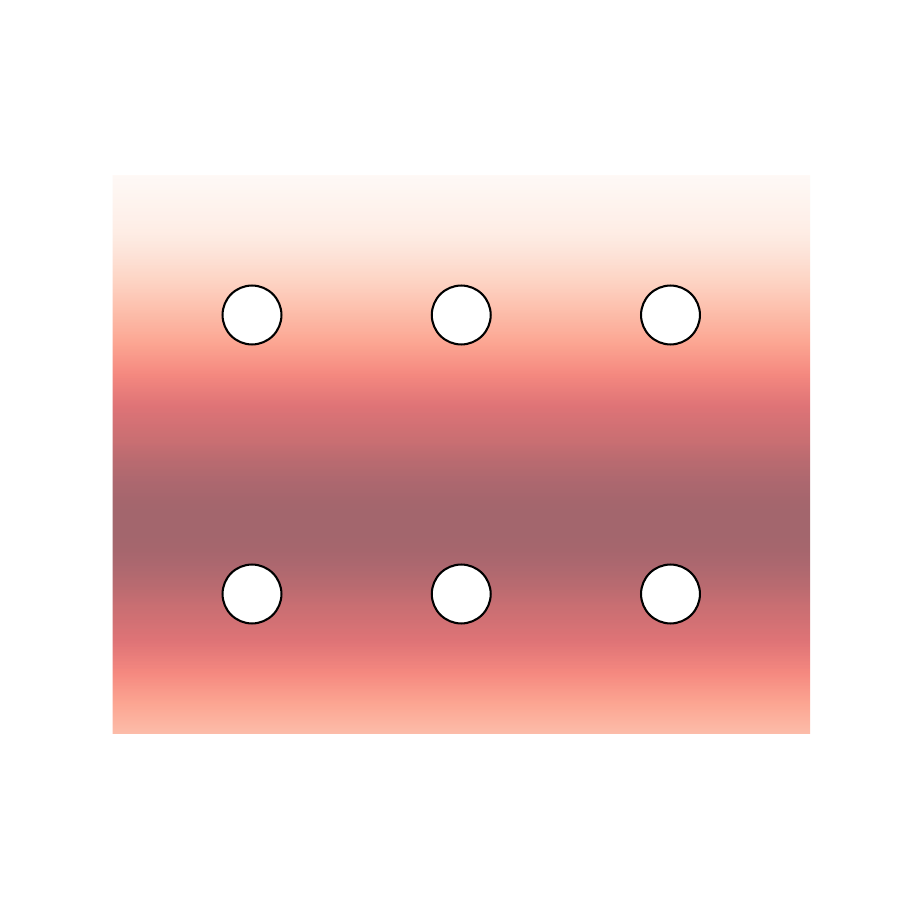}
        \begin{picture}(0,0)
            \put(-50, 100){\makebox(0,0)[lt]{\textbf{(a)}}}
        \end{picture}
    \end{minipage}
    \hfill
    % Second image with label (b) at the top left
    \begin{minipage}[t]{0.4\columnwidth}
        \includegraphics[width=\linewidth]{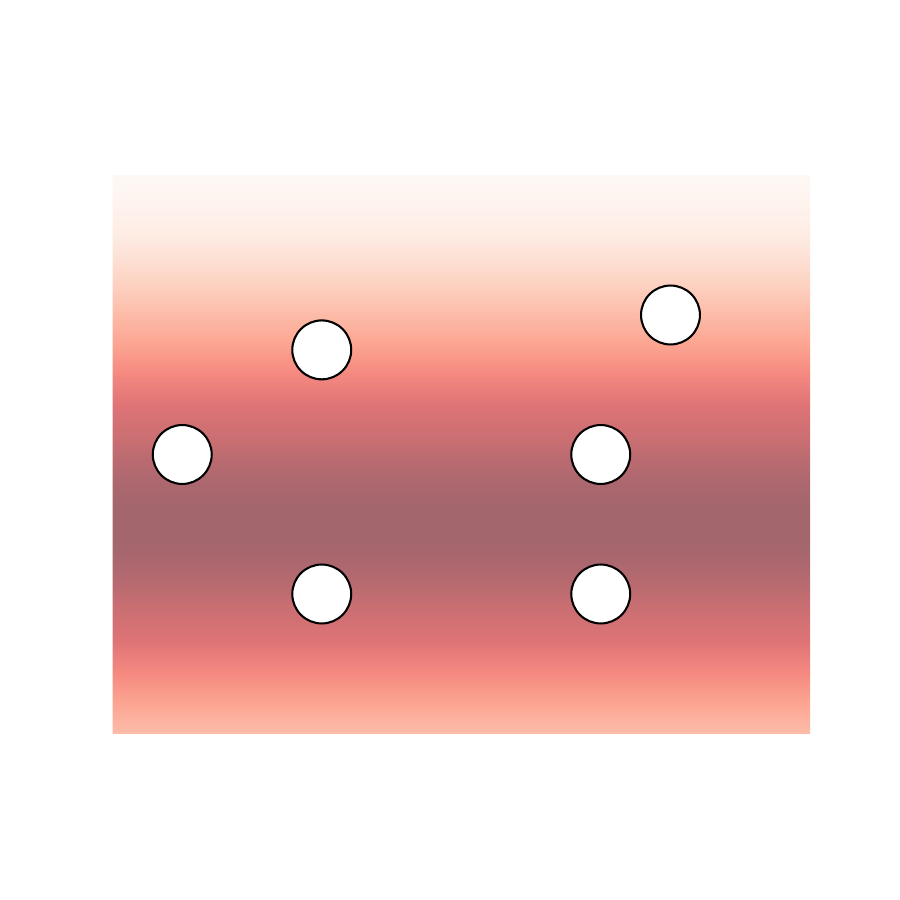}
        \begin{picture}(0,0)
            \put(-50, 100){\makebox(0,0)[lt]{\textbf{(b)}}}
        \end{picture}
    \end{minipage}
    \hfill
    % Caption for the figure
    \caption{Spatial configurations of Rydberg atoms for the ground-state preparation of the CH molecule. (a) A symmetric configuration that cannot be used because the initial state $\ket{{000000}}$ belongs to a different symmetry sector than the ground state of the CH molecule. (b) This non-symmetric configuration allows us to reach the ground state of CH. The atomic coordinates for this configuration ordered by the spin index (in $\mu$m) are: (0, -12), (12, -24), (21, -12), (0, 12), \mbox{(12, 12),} (24, 18).}
    \label{Fig:CH_config}
\end{figure}
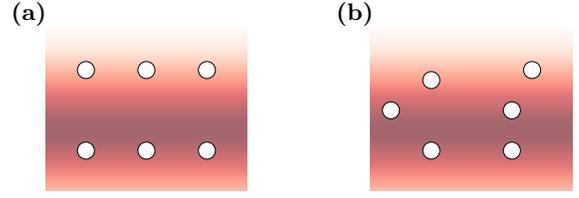
For the optimization, we use the same settings as for the unitary-synthesis case. The maximum number of figure-of-merit evaluations is limited to \(\leq1.5 \times 10^6\) for all methods. For dCRAB, we set the final evolution time to \(T = 2000\,\mu\mathrm{s}\) and the bandwidth limit to \(\Delta\Omega = 200/T\). For RALLY$_\text{A}$, we also use \(T = 2000\,\mu\mathrm{s}\).

\section{Von-Neumann entropy of the ground states of the molecules} \label{app:entropy}
% We quantify entanglement by calculating the von-Neumann entropy $S$ of the reduced density matrix of a single qubit and averaging $S$ over all qubits. 
We quantify entanglement by calculating the average single-qubit von-Neumann entropy of the reduced density matrices $\rho_i$ of each qubit
\begin{equation}
    S = -\frac{1}{n} \sum_{i=1}^n \mathrm{Tr}\!\left[\rho_i \log \rho_i\right],
\end{equation}
where $n$ is the number of qubits.
Table~\ref{table:entanglment} shows the obtained von-Neumann entropies $S$ along with the corresponding symmetry-sector dimensions. The ground states of NO$_3$ and CH exhibit significant higher entanglement than the ground state of H$_2$. 
% The von-Neumann entropy was determined by means of the exact diagonalization of the Hamiltonian operators.

\begin{table}[h]
    \centering
    \begin{tabular}{p{1.5cm}p{1.5cm}p{1.5cm}p{1.5cm}}
        \hline
         & H$_2$ & NO$_3$ & CH  \\
        \hline
         dim & 9 & 16 & 64  \\
        \hline         
        $S$  & 0.096 & 0.818 & 0.352  \\
        \hline
    \end{tabular}
    \caption{Entanglement $S$ and relevant symmetry sector dimension dim of the three molecules analyzed.}
    \label{table:entanglment}
\end{table}

\section{Calculation of gradients}  \label{app:grad}
\subsection{RALLY$_\text{T}$}

To calculate the gradient of a figure of merit for the RALLY$_\text{T}$ method, an analytical expression for the derivative
$ \partial U_{\text{R}}/\partial \tau_{L}$ is required. Let us first introduce the following notation
%-----------------------------------------------
% 1.  Shorthand notation
%-----------------------------------------------
\newcommand{\UpL}[2]{U_{#1}^{(#2)}}      % U_p^(ℓ)
\def\Up{\UpL{p}{\ell}}                   % U_p^(ℓ) with implicit ℓ
\begin{equation}
U^{(\ell,p)}=\exp\left({-i\,\frac{\tau_{\ell}}{N_\text{P}}{\cal H}^{(\ell,p)}}\right) , 
\ \ \ 
U_{\text{R,T}}=
     \prod_{\ell=1}^{N_\text{L}}
     \prod_{p=1}^{N_\text{P}}
     U^{(\ell,p)},
\end{equation}

% Shorthands
\def\Upl{U_{p}^{(\ell)}}  % single ℓ–factor
\def\UpL#1#2{U_{#1}^{(#2)}}   % U_j^(ℓ)
\newcommand{\Uleft}{U_{<L}}    % all blocks with j<p
\newcommand{\Uright}{U_{>L}}   % all blocks with j>p
\newcommand{\Upre}[1]{U_{L}^{(<#1)}}   % ℓ = 1 … r-1
\newcommand{\Upost}[1]{U_{L}^{(\ge #1)}}% ℓ = r … L

\begin{equation}
\Uleft = \prod_{\ell=1}^{L-1}\prod_{p=1}^{N_\text{P}}U^{(\ell,p)},
\ \ \ 
\Uright = \prod_{\ell=L+1}^{N_\text{L}}\prod_{p=1}^{N_\text{P}}U^{(\ell,p)},
\end{equation}

\begin{equation}
\Upre{r} = \prod_{p=1}^{r-1}U^{(L,p)},
\ \ \ 
\Upost{r} = \prod_{p=r}^{N_\text{P}}U^{(L,p)},
\end{equation}
where ${\cal H}^{(\ell,p)}$ is defined in Eq.~\eqref{eqn_hlp}. Using this notation, the derivative reads
\begin{equation}
    \displaystyle
    \frac{\partial U_{\text{R,T}}}{\partial \tau_{L}}
    =\Uleft\;
    \left[
         \sum_{r=1}^{N_\text{P}}
         \frac{(-i)}{N_\text{P}}\,\Upre{r}\,{\cal H}^{(L, r)}\,\Upost{r} 
    \right]\;
    \Uright,
\end{equation}
where we used that all pulses in a layer have the same time duration.

\subsection{GRAPE}
For the GRAPE algorithm, computing the gradient of the unitary propagator with respect to the control amplitudes generally requires the Fr\'echet derivative of the matrix exponential. For time step \(j\) with piecewise-constant control \(u(j)\), the Hamiltonian reads ${\cal H}_j \equiv {\cal H}_0 + u(j)\,{\cal H}_c$, and the corresponding propagator is
\begin{equation}
U_j \equiv U_j(\Delta t) = \exp\!\left(-i\Delta t {\cal H}_j\right).
\end{equation}
A small perturbation \(u(j) \to u(j) + \delta u(j)\) induces, to first order,
\begin{equation}
\delta U_j = - i\Delta t\,\delta u(j)\, \overline{{\cal H}_c}\, U_j,
\end{equation}
where
\begin{equation}
\overline{{\cal H}_c}\, \Delta t
= \int_{0}^{\Delta t} U_j(\tau)\, {\cal H}_c\, U_j(-\tau)\, d\tau,
\end{equation}
which follows from Duhamel’s formula
\begin{equation}
\left.\frac{d}{dx}\, e^{A + x B}\right|_{x=0}
= \int_{0}^{1} e^{(1-s)A}\, B\, e^{sA}\, ds.
\end{equation}
For sufficiently small time steps (\(\Delta t \ll \|{\cal H}_j\|^{-1}\)), the relation \(\overline{{\cal H}_c} \approx {\cal H}_c\) holds such that the gradient can be approximated using only matrix exponentials as
\begin{equation}
\delta U_j \approx - i\Delta t\,\delta u(j)\,{\cal H}_c U_j,
\quad\text{i.e.,}\quad
\frac{\partial U_j}{\partial u(j)} \approx - i\Delta t\,{\cal H}_c U_j.
\end{equation}
However, approximating the Fr\'echet derivative yields inexact gradients which can change the optimization trajectory and, most critically, the probability of success (the fraction of optimizations that reach the target fidelity). We observe this effect in our Ising-spin-chain case study.

Moreover, the small-\(\Delta t\) approximation usually does not become more accurate as the dimensionality of the problem increases.
The minimum evolution time \(T_{\min}\), below which high-fidelity solutions are unattainable, usually increases with the Hilbert-space dimension \(N\) \cite{GSL1,qsl_spinchain}. For the Ising spin chain, we empirically find \(T_{\min} \propto N = 2^{n}\). This has direct implications for the discretization of time for the GRAPE algorithm. Let \(P\) be the number of optimization parameters, and \(\Delta t = T/P\) the pulse-time discretization. Then, when working with a number of parameters close to the lower bound, \(P \gtrsim 2N - 2\), the discretization does not automatically become finer.

\subsection{RALLY$_\text{A}$}
For the propagator $U_{\text{R,A}}(\bm{\xi})$ of RALLY$_\text{A}$ in Eq.~\eqref{eq:RALLY_AMP}, we introduce the pulse Hamiltonian and the single-pulse propagator
\begin{align}
  {\cal H}^{(\ell,p)}(\xi_\ell)
    &= {\cal H}_0 + \xi_\ell\,u^{(\ell,p)}{\cal H}_c,
  \\
  U^{(\ell,p)}(\xi_\ell)
    &= \exp\left[-i\Delta t\,{\cal H}^{(\ell,p)}(\xi_\ell)\right],
\end{align}
so that
\begin{equation}
  U_{\text{R,A}}(\bm{\xi})
  =
  \prod_{\ell=1}^{N_\text{L}}
  \prod_{p=1}^{N_\text{P}}
  U^{(\ell,p)}(\xi_\ell).
\end{equation}
Using the same notation as in the previous sections, the exact derivative with respect to the layer amplitude $\xi_L$ is
\begin{equation}
  \label{eq:RALLY_A_grad_exact}
  \begin{aligned}
    \frac{\partial U_{\text{R,A}}}{\partial \xi_L}
    &=
    U_{<L}
    \biggl[
      \sum_{r=1}^{N_\text{P}}
        \bigl(-i\Delta t\,u^{(L,r)}\bigr)
        U_{L}^{(<r)}\overline{{\cal H}}_{c}^{(L,r)}
        U_{L}^{(\ge r)}
    \biggr]
    U_{>L},
  \end{aligned}
\end{equation}
where $\overline{{\cal H}}_{c}^{(L,r)}$ is the Duhamel (Fr\'echet) integral of the control Hamiltonian
\begin{equation}
  \label{eq:RALLY_A_Hbar}
  \overline{{\cal H}}_{c}^{(L,r)}\,\Delta t
    = \int_0^{\Delta t}
         U^{(L,r)}(\tau)\,
         {\cal H}_c\,
         U^{(L,r)}(-\tau)\,d\tau
\end{equation}
and
\begin{equation}
  U^{(L,r)}(\tau) = \exp\bigl[-i\tau\,{\cal H}^{(L,r)}(\xi_L)\bigr].
\end{equation}

\section{Theoretical estimates for the runtime scaling of RALLY$_\text{T}$ and GRAPE} \label{app_scaling}
For the GRAPE algorithm, a single gradient evaluation is based on the Fréchet derivative of the matrix exponential, which can be obtained from matrix diagonalization exhibiting a $\mathcal{O}(N^{3})$ scaling (see Appendix~\ref{app:grad} for details). For a control field, at least \(2N-2\) time bins (optimization parameters) are required for ground-state preparation and state-transfer problems and, hence, one optimization iteration of GRAPE involves \(2N-2\) diagonalizations. For \(n_{\text{it}}\) optimization iterations, this yields for the runtime scaling
\begin{equation} \label{eqn_c_grape}
\mathcal{C}_{\text{GRAPE}}=\mathcal{O}\bigl(n_{\text{it}}N^{4}\bigr).
\end{equation}
Note that one can avoid calculating the full Fréchet matrix by directly using Krylov-subspace methods, which typically reduces the runtime scaling to \(\mathcal{O}(n_{\mathrm{it}}N^{3})\). However, for the system sizes considered in the main text, using the Krylov-subspace method resulted in an overall slower calculation and, therefore, full diagonalization was used.

In the RALLY$_\text{T}$ method, diagonalization is performed only once prior to the optimization; for example, diagonalization needs to be performed for the two static Hamiltonians with the control values \(J=\pm 1\) in Sec.~\ref{sec_ising}. This pre-diagonalization exhibits an \(\mathcal{O}(N^3)\) scaling of the runtime. For the optimization, each evaluation of the figure of merit and its gradient involves $N_\text{P}N_\text{L}$ matrix–vector multiplications. For $N_\text{L}\approx 2N-2$ and a fixed layer size $N_\text{P}$, the scaling for $n_{\text{it}}$ iterations hence reads
\begin{equation} \label{eqn_c_rally}
\mathcal{C}_{\rm RALLY_{\text{T}}}
= \underbrace{\mathcal{O}(N^{3})}_{\text{pre-diagonalization}}
+ \underbrace{\mathcal{O}(n_{\text{it}}\,N^{3})}_{\text{per-iteration cost}},
\end{equation}
where the second term usually dominates. From practical calculations, we observe that the number of iterations $n_{\text{it}}$ necessary for finding the solution scales approximately as $n_{\text{it}}=\mathcal{O}(N^{1/2})$ for both RALLY methods and GRAPE.

Comparing the runtime scaling for GRAPE in Eq.~\eqref{eqn_c_grape} and for RALLY$_{\text{T}}$ in Eq.~\eqref{eqn_c_rally} implies that the RALLY$_\text{T}$ method exhibits a more favorable scaling than GRAPE for the system sizes considered.

\section{Finite-bandwidth interpolation between control pulses} \label{app:sigmoid}
To account for finite control bandwidth in experimental settings, we introduce a smooth interpolation for each discontinuous jump between the constant-amplitude pulses in the RALLY$_\text{T}$ propagator $U_\text{R,T}$. When the control is changed from the starting pulse amplitude $u_0$ to the final pulse amplitude $u_1$ over a time interval $[t_0,t_1]$ of duration $\tau_{\mathrm{rise}} = t_1 - t_0$, we use the interpolation $f_{\mathrm{int}}(t)=u_0+ (u_1-u_0)s(t)$ with $t \in [t_0,t_1]$ and the shifted sigmoid function
\begin{equation}
  s(t)
  = \frac{1}{2}\left[
      1 + \tanh\!\left(\frac{k}{2}\,(t - t_\text{m})\right)
    \right],
\end{equation}
where
\begin{equation}
  k 
  = \frac{2}{\tau_{\mathrm{rise}}}
    \ln\!\left(\frac{1 - \varepsilon}{\varepsilon}\right),
  \qquad 
  t_\text{m} = \frac{t_0 + t_1}{2}.
\end{equation}
The time $\tau_{\mathrm{rise}}$ determines the duration of the transition centered at $t_\text{m}$, while $\varepsilon$ controls the shape. The steepness $k$ is fixed by $\tau_{\mathrm{rise}}$ and $\varepsilon$ so that $f_{\mathrm{int}}(t_\text{m} \pm \tau_{\mathrm{rise}}/2)$ is $\varepsilon$ away from $u_1$ and $u_0$, respectively.

For the numerical simulations, we divide $[t_0,t_1]$ into $N_{\mathrm{int}}$ equal intervals of length $\Delta t = \tau_{\mathrm{rise}} / N_{\mathrm{int}}$ and approximate the continuous interpolation by a product of short-time propagators with constant Hamiltonian in these intervals. We define $u_n = f_{\mathrm{int}}(t_n)$ for $t_n = t_0 + (n - 1/2)\Delta t$ and $n = 1,\dots,N_{\mathrm{int}}$. Using ${\cal H}(u_n) = {\cal H}_0 + u_n{\cal H}_c$, the unitary propagator associated with a single interpolation reads
\begin{equation}
  U_{\mathrm{rise}}(u_0 \to u_1)
  = \prod_{n=1}^{N_{\mathrm{int}}}
    \exp\biggl[-i\,\Delta t\,{\cal H}(u_n)\biggr].
\end{equation}
In the RALLY$_\text{T}$ propagator $U_\text{R,T}$, an interpolation propagator $U_{\mathrm{rise}}$ is
inserted between each pair of constant-amplitude pulses, so that each instantaneous jump $u_0 \to u_1$ is replaced by a smooth
finite-time interpolation. This yields a pulse sequence with finite-bandwidth pulses. Importantly, these interpolation propagators can be prepared prior to the optimization for RALLY$_\text{T}$, because the pulse amplitudes are randomly chosen in advance and the interpolation does not depend on the pulse durations which are varied. This reduces the computational overhead for including finite control bandwidth in the optimization significantly.

For the specific optimization problem studied in Sec.~\ref{sec_ising}, we chose $\tau_{\mathrm{rise}} = 10$, $N_{\mathrm{int}} = 100$ and $\varepsilon = 10^{-10}$ for each interpolation propagator.

\bibliography{bibliography}

\end{document}